\documentclass[12pt,a4paper]{article}

\usepackage[utf8]{inputenc}

\usepackage{amsmath, amssymb, amsthm, amsfonts, amsbsy, mathtools}
\usepackage{mathrsfs, bm, bbm, stmaryrd, relsize}
\usepackage{graphicx, tikz, empheq}
\usepackage{enumerate, fullpage}
\usepackage{titlesec}

\DeclareMathAlphabet{\mathds}{U}{dsrom}{m}{n}

\newtheorem{thm}{Theorem}[section]
\newtheorem{cor}[thm]{Corollary}
\newtheorem{lem}[thm]{Lemma}

\newtheorem{defn}[thm]{Definition}
\newtheorem{rem}[thm]{Remark}

\titleformat{\section}
  {\normalfont\large\bfseries}{\thesection}{1em}{}
\titlespacing*{\section}{0pt}{2ex plus 0.5ex minus 0.2ex}{1ex}

\titleformat{\subsection}
  {\normalfont\normalsize\bfseries}{\thesubsection}{1em}{}
\titlespacing*{\subsection}{0pt}{1.5ex plus 0.3ex minus 0.1ex}{0.8ex}

\allowdisplaybreaks

\makeatletter
\renewcommand\@makefnmark{}
\makeatother

\renewenvironment{abstract}{
  \begin{center}
  \vspace{-0.5em}
  \rule{\textwidth}{0.4pt}\\[4pt]
  \begin{minipage}{0.95\textwidth}
  \small
}{
  \end{minipage}\\[4pt]
  \rule{\textwidth}{0.4pt}
  \end{center}
  \vspace{1em}
}

\numberwithin{equation}{section}

\makeatletter
\renewcommand{\maketitle}{
\begin{center}
    {\Large\textbf{Kinetic Theory and the Mechanics of Isothermal Gas Spheres:
    Derivation of the Classical Emden--Chandrasekhar Equation
    via the Vlasov--Poisson Formalism}}
    \vskip 1em

    {\normalsize \textsc{Steven D. Miller}\footnote{\texttt{stevendm@ed-alumni.net}}}
    \vskip 0.5em
\end{center}
}
\makeatother

\date{} 
\begin{document}

\maketitle

\begin{abstract}
\noindent We present a derivation of the mechanics of isothermal gas spheres directly from the Vlasov--Poisson equation. By extremising the Boltzmann entropy, we obtain the Maxwell--Boltzmann distribution for a self-gravitating Newtonian gas, which is a stationary solution of the Vlasov--Poisson system. From this distribution, the corresponding Poisson--Boltzmann equation for the gravitational potential is deduced. The second variation of entropy reproduces the classical Antonov instability criterion: the critical energy is $E_c \simeq -0.335\,\frac{G M^2}{R}$, below which no local entropy maximum exists and the configuration becomes unstable (the so-called "gravothermal catastrophe"). In this work, we assume $E>E_c$, so all equilibria lie on the stable branch, and the Antonov instability does not affect the analysis. Specializing to spherical symmetry, we recover the classical equation for the isothermal gas sphere, as originally studied by Chandrasekhar, which has applications to the formation of red giant stars. We also derive the fundamental equation of hydrostatic equilibrium, the energy integral and the virial theorem directly from the stationary Vlasov--Poisson solution, demonstrating also that an isothermal gas exhibits negative specific heat. Furthermore, we show that an isothermal gas sphere of strictly constant density is an impossibility. This exposition emphasizes some of the deep connections and self-consistency between kinetic theory, statistical mechanics, and stellar structure, while highlighting formal mathematical aspects of classical astrophysical models.
\end{abstract}

\tableofcontents

\section{Isothermal gas sphere mechanics: introduction and historical
review}
\begin{quote}
\textit{I write not to teach others but to inform myself. What I write here is not my teaching but my study; it is not a lesson for others but for myself...yet if it happens that others take profit from it, let them take it.} \\
--- Michel de Montaigne, \textit{Essai}, Book II, Chapter 18, ``De L’affection Des Pères Aux Enfants'', 1580
\end{quote}
\noindent The study of self-gravitating systems, such as stars, globular clusters, and galaxies, has been a central topic in astrophysics since the early twentieth century. These systems are governed by long-range Newtonian interactions, leading to rich thermodynamic and dynamical behaviors that distinguish them from ordinary short-range systems. A particularly elegant framework for describing their evolution is provided by the \emph{Vlasov equation}, introduced by Anatoly Vlasov in 1938 \textbf{[1]}. This kinetic equation governs the time evolution of the one-particle distribution function $\mathcal{F}(\mathbf{x},\mathbf{v},t)$ in phase space under the self-consistent gravitational or electrostatic potential generated by the system itself. It is valid in the collisionless regime, appropriate for systems in which the evolution is dominated by long-range interactions rather than short-range collisions, and has been extensively applied to problems in astrophysics, cosmology, and plasma physics \textbf{[2--14]}.

In the context of a self-gravitating hydrogen gas, such as a diffuse molecular cloud, the Vlasov--Poisson formalism can be rigorously justified starting from first principles. Beginning with the Liouville equation for an $N$-body system interacting via the Newtonian potential, one obtains the BBGKY hierarchy \textbf{[15--18]}, which describes the evolution of $s$-particle distribution functions. In the mean-field limit, where $N \to \infty$ and two-body correlations become negligible compared with the collective gravitational field, the hierarchy closes at the level of the one-particle distribution function, yielding the Vlasov--Poisson system. Each particle then evolves in the smooth gravitational potential generated by the ensemble of all others. In the mean-field scaling, each pairwise potential is effectively taken as $1/N$ times the Newtonian potential, ensuring that the total force per particle remains finite as $N \to \infty$. This scaling preserves the extensive character of the total potential energy while producing a smooth mean-field potential in the continuum limit. \emph{This limit is applicable when the system is sufficiently dilute and large-scale so that gravitational interactions dominate over short-range collisions}. The Maxwell--Boltzmann distribution is a stationary solution in this regime. 

The validity of this mean-field reduction has been rigorously established in mathematical physics: Braun and Hepp \textbf{[19]} proved convergence of the empirical one-particle distribution to the Vlasov solution in the large-$N$ limit, Spohn \textbf{[20]} provided a detailed analysis of the limit and its implications for kinetic theory, and Binney and Tremaine \textbf{[21]} discussed the physical relevance in astrophysical systems. Together, these results justify the Vlasov--Poisson equation as a fundamental continuum description of large self-gravitating particle ensembles.

In the context of stellar structure, the \emph{isothermal gas sphere} model, first introduced by Emden \textbf{[25]} in 1907 and later developed extensively by Chandrasekhar \textbf{[26,27]}, provides a simplified but remarkably instructive description of a self-gravitating, spherically symmetric gas in hydrostatic equilibrium. Assuming an isothermal equation of state, the density profile satisfies a nonlinear Poisson equation whose dimensionless form is now known as the Emden--Chandrasekhar equation. The isothermal sphere has served as a classical model for stellar envelopes, globular clusters, and even the evolution of red giants \textbf{[28-36]}.

Beyond its astrophysical applications, the isothermal sphere occupies a central place in the \emph{statistical mechanics of self-gravitating systems}. Because gravity is a long-range, attractive, and non-extensive interaction, the usual thermodynamic framework breaks down: the canonical and microcanonical ensembles are no longer equivalent, specific heat can become negative, and stable equilibria may cease to exist below a critical energy. These features were first analyzed in a seminal paper by Antonov (1962), who investigated the thermodynamic stability of an isothermal self-gravitating gas confined within a spherical boundary \textbf{[37]}. By examining the second variation of the Boltzmann entropy under fixed mass and energy, Antonov showed that equilibrium configurations exist only for total energies $E > E_c \simeq -0.335\, G M^2 / R$, and that below this critical energy, no local entropy maximum is possible---a phenomenon now known as the \emph{Antonov instability} or \emph{gravothermal catastrophe}. Lynden-Bell \textbf{[38]} later extended this framework to include the concept of \emph{violent relaxation}, providing a statistical interpretation of how collisionless systems can approach quasi-equilibrium states described by the Vlasov equation. In this work, we always assume $E>E_c$, so that all equilibria lie on the stable branch, and the Antonov instability is not relevant.

The statistical mechanics of self-gravitating Newtonian gases therefore involves subtle interplays between dynamics, thermodynamics, and stability. Unlike conventional gases, self-gravitating systems can exhibit entropy maxima that are only \emph{metastable}, and their evolution can depend sensitively on the chosen ensemble (microcanonical versus
canonical). These difficulties highlight the need for a consistent kinetic-theoretic approach, where equilibrium distributions arise naturally from the underlying Vlasov dynamics rather than being postulated thermodynamically. These issues have been discussed in a number of works \textbf{[39--49]}.

The purpose of the present work is to derive the mechanics of isothermal gas spheres from first principles, starting directly from the stationary, collisionless Vlasov--Poisson equation. By maximizing the Boltzmann entropy under fixed mass and energy, we recover the Maxwell--Boltzmann distribution for a self-gravitating gas, which is shown to be a stationary solution of the Vlasov equation. From this distribution, the corresponding Poisson--Boltzmann equation for the gravitational potential follows immediately. Specializing to spherical symmetry, we obtain the classical Emden--Chandrasekhar equation for the isothermal gas sphere. We further show that the second variation of the entropy reproduces the classical Antonov stability criterion, thereby connecting the statistical-mechanical approach to the mean-field kinetic description. This unified derivation clarifies the conceptual and mathematical relationship between the Vlasov framework, entropy extremization, and the theory of isothermal spheres, linking kinetic theory, statistical mechanics, and stellar structure in a single coherent formalism.
\subsection{The classical Emden-Chandrasekhar isothermal gas sphere}
We begin with the fundamental equation of hydrostatic equilibrium for a self-gravitating Newtonian gas in a closed domain $\Omega\subset\mathbb{R}^{3}$ in general coordinates
\begin{equation}
\nabla P(\mathbf{x}) = - \rho(\mathbf{x}) \nabla \Phi(\mathbf{x}),
\end{equation}
where $P(\mathbf{x})$ is the pressure, $\rho(\mathbf{x})$ the density, and $\Phi(\mathbf{x})$ the gravitational potential. For an isothermal ideal gas at temperature T, the equation of state reads
\begin{equation}
P(\mathbf{x}) = \frac{k_B T}{m} \rho(\mathbf{x}),
\end{equation}
so that
\begin{equation}
\nabla \left( \frac{k_B T}{m} \rho(\mathbf{x}) \right) = - \rho(\mathbf{x}) \nabla \Phi(\mathbf{x})
\end{equation}
and
\begin{equation}
\frac{k_B T}{m} \nabla \rho(\mathbf{x}) = - \rho(\mathbf{x}) \nabla \Phi(\mathbf{x})
\end{equation}
Dividing both sides by $\rho(\mathbf{x}) > 0$, we obtain
\begin{equation}
\frac{\nabla \rho(\mathbf{x})}{\rho(\mathbf{x})} = - \frac{m}{k_B T} \nabla \Phi(\mathbf{x})
\end{equation}
so that
\begin{equation}
\frac{\nabla \rho(\mathbf{x})}{\rho(\mathbf{x})}\equiv \nabla \ln \rho(\mathbf{x}) = - \frac{m}{k_B T} \nabla \Phi(\mathbf{x}).
\end{equation}
This is an exact differential, which can be integrated along any path in $\mathbf{x}$-space or $\Omega$ from a reference point $\mathbf{x}_0$ (typically the center of the configuration) to an arbitrary point $\mathbf{x}$:
\begin{equation*}
\int_{\mathbf{x}_0}^{\mathbf{x}} \nabla \ln \rho(\mathbf{x}') \cdot d\mathbf{x}'
= - \frac{m}{k_B T} \int_{\mathbf{x}_0}^{\mathbf{x}} \nabla \Phi(\mathbf{x}') \cdot d\mathbf{x}'.
\end{equation*}
The left-hand side integrates to
\begin{equation*}
\ln \rho(\mathbf{x}) - \ln \rho(\mathbf{x}_0) = \ln \frac{\rho(\mathbf{x})}{\rho_0},
\end{equation*}
where $\rho_0 = \rho(\mathbf{x}_0)$, while the right-hand side gives
\begin{equation*}
- \frac{m}{k_B T} (\Phi(\mathbf{x}) - \Phi(\mathbf{x}_0)).
\end{equation*}
Hence
\begin{equation}
\rho(\mathbf{x})=\rho_{o}\exp\left(-\frac{m}{k_{B}T}\Phi(x)\right)\exp\left(-\frac{m}{k_{B}T}\Phi(\mathbf{x}_{o})
\right)
\end{equation}
Absorbing $\exp(-\frac{m}{k_{B}T}\Phi(\mathbf{x}_0)$ into the constant $\rho_0$, we then obtain the density distribution:
\begin{equation}
\rho(\mathbf{x}) = \rho_0 \exp\Big(-\frac{m}{k_B T} \Phi(\mathbf{x})\Big)
\end{equation}
Substituting this into Poisson's equation for the gravitational potential,
\begin{equation}
\Delta \Phi(\mathbf{x}) = 4 \pi G \rho(\mathbf{x}) = 4 \pi G \rho_0 \exp\Big[-\frac{m}{k_B T} \Phi(\mathbf{x})\Big],
\end{equation}
yields the general-coordinate form of the Poisson-Boltzmann equation.

Finally, assuming spherical symmetry $\Phi = \Phi(r)$, the Laplacian reduces to
\begin{equation*}
\Delta \Phi(r) = \frac{1}{r^2} \frac{d}{dr} \Big( r^2 \frac{d \Phi(r)}{dr} \Big),
\end{equation*}
so that the equation becomes the classical Emden-Chandrasekhar equation for an isothermal gas sphere:
\begin{equation*}
\frac{1}{r^2} \frac{d}{dr} \Big( r^2 \frac{d \Phi(r)}{dr} \Big) = 4 \pi G \rho_0 \exp\Big[- \frac{m}{k_B T} \Phi(r) \Big].
\end{equation*}
An \emph{isothermal gas sphere} is a spherically symmetric, self-gravitating gas distribution in hydrostatic equilibrium, with constant temperature $T$. The pressure $P$ and density $\rho$ satisfy the ideal gas law. For spherical symmetry, hydrostatic equilibrium reads
\begin{equation*}
\frac{dP}{dr} = - \rho(r) \frac{G M(r)}{r^2}, \quad
M(r) = \int_0^r 4 \pi r'^2 \rho(r') \, dr'.
\end{equation*}
Introducing the dimensionless potential
\begin{equation}
\Psi(\mathbf{x})= \frac{m}{k_B T} \big(\Phi(\mathbf{x}) - \Phi(0)\big),
\end{equation}
we have
\begin{align}
&\Phi(\mathbf{x}) = \Phi(0) + \frac{k_B T}{m} \Psi(\mathbf{x})\\&
\Delta \Phi(\mathbf{x}) = \frac{k_B T}{m} \Delta \Psi(\mathbf{x})
\end{align}
Substituting into Poisson's equation gives the dimensionful Emden-Chandrasekhar equation,
\begin{equation}
\Delta \Psi(\mathbf{x}) = \frac{4 \pi G m}{k_B T} \rho_0 e^{-\Psi(\mathbf{x})}.
\end{equation}
Assuming spherical symmetry, this reduces to the ordinary differential equation
\begin{equation}
\frac{1}{r^2} \frac{d}{dr} \Big( r^2 \frac{d \Psi(r)}{dr} \Big) = \frac{4 \pi G m}{k_B T} \rho_0 e^{-\Psi(r)},
\end{equation}
with boundary conditions
\begin{equation}
\Psi(0) = 0, \quad \frac{d \Psi(0)}{dr}(0) = 0.
\end{equation}
The density profile is
\begin{equation}
\rho(r) = \rho_0 \, e^{-\Psi(r)}.
\end{equation}
Substituting into Poisson's equation yields the classical \emph{Emden--Chandrasekhar} equation for an isothermal sphere \textbf{[26,27]}:
\begin{equation}
\frac{1}{r^2} \frac{d}{dr} \Big( r^2 \frac{d \Psi(r)}{dr} \Big)
= 4 \pi G \frac{m}{k_B T} \, \rho_0 \, e^{-\Psi(r)}.
\end{equation}
with boundary conditions $\Psi(0) = 0 $ and $\frac{d\Psi(r)}{dr}\Big|_{r=0} = 0 $, ensuring regularity at the center. Introducing the dimensionless radius
\begin{equation}
\xi = \frac{r}{r_0}, \qquad r_0 = \sqrt{\frac{k_B T}{4 \pi G m \rho_0}},
\end{equation}
the Emden--Chandrasekhar equation becomes
\begin{equation}
\frac{1}{\xi^2} \frac{d}{d\xi} \left( \xi^2 \frac{d \Psi(\xi)}{d\xi} \right) = e^{-\Psi(\xi)},
\end{equation}
which corresponds to a Lane--Emden equation of index $n = \infty$ or $\gamma=1$,\textbf{[26]}.

\subsection{Basic properties of the isothermal sphere}
Consider the dimensionful Chandrasekhar-Emden equation in spherical symmetry with the prescribed boundary conditions
\begin{equation}
\frac{1}{r^2}\frac{d}{dr}\Bigl(r^2 \frac{d \Psi(r)}{dr} \Bigr) = 4 \pi G \frac{m}{k_B T} \rho_0 e^{-\Psi(r)},
\qquad \Psi(0)=0, \quad \frac{d\Psi(r)}{dr}(0)=0,
\end{equation}
where $\Psi(r) = \frac{m}{k_B T} (\Phi(r)-\Phi(0))$ is the dimensionless potential and $\rho_0$ the central density.
\noindent \underline{Small-$r$ behavior}: Near the center, assume a quadratic ansatz
\begin{equation}
\Psi(r) = A r^2.
\end{equation}
Then
\begin{equation*}
\frac{d\Psi}{dr} = 2 A r, \qquad
\frac{d}{dr}\left(r^2 \frac{d\Psi}{dr}\right) = \frac{d}{dr} (2 A r^3) = 6 A r^2 \sim 6A
\end{equation*}
so that the left-hand side becomes $6 A$. Expanding the exponential on the right-hand side, $e^{-A r^2} \approx 1$ for small $r$, we find
\begin{equation*}
6 A \approx 4 \pi G \frac{m}{k_B T} \rho_0 \quad \implies \quad
A = \frac{2 \pi G m \rho_0}{3 k_B T}.
\end{equation*}
\begin{equation}
\Psi(r) \approx \frac{2 \pi G m \rho_0}{3 k_B T} r^2.
\end{equation}

\subsubsection{Asymptotic behavior of the isothermal sphere}
It is important to consider the asymptotic properties of the isothermal sphere for small and large $r$. Consider the dimensionless Emden--Chandrasekhar equation for an isothermal gas sphere,
\begin{equation}
\frac{1}{r^2}\frac{d}{dr}\Big(r^2 \frac{d\Psi(r)}{dr}\Big) = e^{-\Psi(r)}.
\end{equation}
\begin{enumerate}
\item The finite-mass case: If the right-hand side is negligible (e.g., $\rho(r) \to 0$ sufficiently rapidly at large $r$), the equation reduces to Laplace's equation,
\begin{equation*}
\frac{1}{r^2}\frac{d}{dr}\Big(r^2 \frac{d\Psi}{dr}\Big) = 0.
\end{equation*}
so that $\Psi(r)$ is harmonic. Integrating twice gives
\begin{equation*}
r^2 \frac{d\Psi}{dr} = C_1 \;\Rightarrow\; \frac{d\Psi}{dr} = \frac{C_1}{r^2} \;\Rightarrow\; \Psi(r) = -\frac{C_1}{r} + C_2,
\end{equation*}
which is the familiar Newtonian potential of a finite-mass distribution.
\item Isothermal sphere: For an isothermal gas, the density decays slowly:
\begin{equation*}
\rho(r) = \rho_c e^{-\Psi(r)} \sim r^{-2} \quad \text{as } r \to \infty.
\end{equation*}
Thus the right-hand side of the Emden equation cannot be neglected at large $r$:
\begin{equation*}
\frac{1}{r^2}\frac{d}{dr}\Big(r^2 \frac{d\Psi}{dr}\Big) \sim r^{-2}.
\end{equation*}
Assume an asymptotic form
\begin{equation*}
\Psi(r) \sim A \ln r + B.
\end{equation*}
Then
\begin{equation*}
\frac{d\Psi}{dr} = \frac{A}{r}, \qquad \frac{d}{dr}\big(r^2 \frac{d\Psi}{dr}\big) = \frac{d}{dr}(A r) = A,
\end{equation*}
so the left-hand side of the Emden equation behaves as
\begin{equation*}
\frac{1}{r^2}\frac{d}{dr}\big(r^2 \frac{d\Psi}{dr}\big) \sim \frac{A}{r^2}.
\end{equation*}
Equating powers with the right-hand side, $e^{-\Psi}=e^{-A\ln r+B}$, gives for large $r$
\begin{equation*}
\frac{A}{r^2} \sim r^{-2}\sim r^{-A} \quad \Rightarrow \quad A = 2.
\end{equation*}
Hence the isothermal gas sphere has the asymptotic behavior
\begin{equation*}
\Psi(r) \sim 2 \ln r + B, \qquad
\rho(r) \sim \rho_c r^{-2}.
\end{equation*}

\item Physical interpretation: The $1/r$ potential occurs only if the density decays fast enough for the RHS to be negligible. For the isothermal sphere, the slow $r^{-2}$ decay of $\rho(r)$ keeps the right-hand side significant, leading to \emph{logarithmic growth} of $\Psi$ instead of a $1/r$ fall-off. Consequently, the total mass diverges:
\begin{equation*}
M=M(R) = \int_0^R 4 \pi r^2 \rho(r) \, dr.
\end{equation*}
\begin{equation*}
M=M(R) \sim \int_0^R 4 \pi r^2 \cdot \frac{1}{r^2} \, dr = \int_0^R 4 \pi \, dr = 4 \pi R.
\end{equation*}
\begin{equation*}
\lim_{R \to \infty} M(R) = \lim_{R \to \infty} 4 \pi R = \infty.
\end{equation*}
Hence, the classical isothermal sphere has infinite total mass and so an \emph{isolated} isothermal gas sphere cannot exist. In practice, one truncates the sphere at a finite radius or uses models with faster-decaying densities. Near the center: $\Psi(r) \sim r^2$ (satisfying the boundary conditions). Far from the center: $\Psi(r) \sim 2 \ln r + \text{const}$ (giving slow density decay). The total mass then diverges, highlighting the idealized nature of the infinite isothermal sphere.
\end{enumerate}

\subsection{Bonnor--Ebert Spheres}
A \emph{Bonnor--Ebert sphere} \textbf{[50-55]} is a self-gravitating, isothermal gas sphere in hydrostatic equilibrium, truncated by an external pressure $P_{\rm ext}$. It generalizes the classical isothermal sphere by imposing a finite radius and mass. The sphere satisfies the usual criteria for hydrostatic equilibrium
\begin{equation}
\frac{dP}{dr} = - \rho(r) \frac{G M(r)}{r^2}, \qquad
M(r) = 4 \pi \int_0^r \rho(r') r'^2 dr',
\end{equation}
with $P(r) = \rho(r) k_B T / m$ for an isothermal gas of temperature $T$. The density may be expressed as
\begin{equation}
\rho(r) = \rho_c e^{-\Psi(r)}, \qquad
\Psi(r) = \frac{m}{k_B T} \big(\Phi(r) - \Phi(0)\big),
\end{equation}
so that the Poisson equation becomes
\begin{equation}
\frac{1}{r^2} \frac{d}{dr} \Big( r^2 \frac{d\Psi}{dr} \Big) = 4 \pi G \rho_c e^{-\Psi(r)}.
\end{equation}
Truncation by external pressure defines a finite radius $R$:
\begin{equation}
P(R) = P_{\rm ext} \quad \implies \quad \rho(R) = \frac{m P_{\rm ext}}{k_B T},
\end{equation}
giving a finite total mass
\begin{equation}
M = 4 \pi \int_0^R \rho(r) r^2 dr.
\end{equation}
A dimensionless form is obtained by introducing $\xi = r / r_0$ with
\begin{equation}
r_0 = \sqrt{\frac{k_B T}{4 \pi G m \rho_c}},
\quad \text{yielding} \quad
\frac{1}{\xi^2} \frac{d}{d\xi} \Big( \xi^2 \frac{d\Psi}{d\xi} \Big) = e^{-\Psi},
\quad 0 \le \xi \le \xi_{\rm out},
\end{equation}
with boundary conditions $\Psi(0)=0$, $\Psi'(0)=0$, and \(\Psi(\xi_{\rm out}) = -\ln(\rho_{\rm ext}/\rho_c)\). The sphere is stable to radial perturbations if
\begin{equation*}
\xi_{\rm out} < \xi_{\rm crit} \approx 6.45,
\end{equation*}
otherwise it is gravitationally unstable, an expression of the gravothermal catastrophe. Bonnor--Ebert spheres provide finite-mass, stable models of self-gravitating isothermal gas under external pressure. They are widely used \textbf{[50-55]}to model molecular cloud cores, protostellar condensations, and stellar cores in red giants. In the limit of vanishing external pressure, the Bonnor--Ebert sphere reduces to the classical, infinite-mass isothermal sphere, recovering the logarithmic density profile $\rho \sim 1/r^2$ at large radii discussed previously.

\subsection{Connection to Polytropes and the Lane--Emden Equation}
The isothermal gas sphere equation is closely related to the theory of \emph{polytropes} and the \emph{Lane--Emden equation}. A polytrope is a self-gravitating fluid in which the pressure $P$ is related to the density $\rho$ by
\begin{equation}
P = K \rho^{1 + 1/n},
\end{equation}
where $K$ is a constant and $n$ is the \emph{polytropic index}. For a spherically symmetric polytrope, combining hydrostatic equilibrium
\begin{equation}
\frac{dP}{dr} = - \rho \frac{G M(r)}{r^2}, \qquad M(r) = \int_0^r 4 \pi r^2 \rho(r) \, dr,
\end{equation}
with the polytropic equation of state leads to the classical \emph{Lane--Emden equation}\textbf{[26,56]}.
\begin{equation}
\frac{1}{\xi^2} \frac{d}{d\xi} \left( \xi^2 \frac{d \theta}{d\xi} \right) = - \theta^n,
\end{equation}
where $\theta(\xi)$ is the dimensionless density and $\xi$ is the dimensionless radius.

We can consider the isothermal limit of polytropes. An isothermal gas has constant temperature $T$, so again the ideal-gas law gives
\begin{equation}
P = \frac{k_B T}{m} \rho.
\end{equation}
Comparing with the polytropic relation
\begin{equation}
P = K \rho^{1 + 1/n},
\end{equation}
we see that the linear dependence on $\rho$ corresponds to
\begin{equation}
1 + \frac{1}{n} = 1 \quad \Rightarrow \quad n \to \infty,
\end{equation}
or equivalently, $\gamma = 1$. Hence an isothermal gas sphere is the $n \to \infty$ limit of a polytrope. Starting from the dimensionless Lane--Emden equation for polytropes:
\begin{equation}
\frac{1}{\xi^2} \frac{d}{d\xi}\left(\xi^2 \frac{d\theta}{d\xi}\right) = - \theta^n.
\end{equation}
Define a new function $\Psi(\xi)$ via
\begin{equation}
\theta(\xi) = e^{-\Psi(\xi)/n}.
\end{equation}
Then
\begin{equation}
\theta^n = \left(e^{-\Psi/n}\right)^n = e^{-\Psi}.
\end{equation}
Substituting into the Lane--Emden equation and taking the limit $n \to \infty$ gives the \emph{isothermal Lane--Emden equation} (Emden--Chandrasekhar equation):
\begin{equation}
\frac{1}{\xi^2} \frac{d}{d\xi} \left( \xi^2 \frac{d\Psi}{d\xi} \right) = e^{-\Psi}.
\end{equation}
Summary:
\begin{itemize}
    \item Polytrope: $P = K \rho^{1 + 1/n}$, Lane--Emden equation $\frac{1}{\xi^2} (\xi^2 \theta')' = - \theta^n$.
    \item Isothermal limit: $P \propto \rho$, $n \to \infty$, $\gamma = 1$, and the Lane--Emden equation reduces to $\frac{1}{\xi^2} (\xi^2 \Psi')' = e^{-\Psi}$.
\end{itemize}

\subsection{Homology and Scale Invariance}
There is a standard homology (scale) invariance of the Emden–Chandrasekhar equation \textbf{[26]}whereby the scaling changes the length and density scales.
Start from the dimensionless Emden–Chandrasekhar equation for an isothermal sphere:
\begin{equation}
\frac{1}{\xi^{2}}\frac{d}{d\xi}\!\Big(\xi^{2}\frac{d\Psi}{d\xi}\Big)=e^{-\Psi(\xi)}.
\end{equation}
Suppose \(\Psi(\xi)\) is a solution. Define, for a constant \(A>0\),
\begin{equation}
\widetilde{\Psi(\xi)}:=\Psi(A\xi)-2\ln A.
\end{equation}
We show that \(\widetilde{\Psi}\) satisfies the same equation.

Let \(u=A\xi\). Then
\begin{equation}
\frac{d}{d\xi}\Psi(A\xi)=A\frac{d\Psi}{du}, \qquad
\frac{d^2}{d\xi^2}\Psi(A\xi)=A^2\frac{d^2\Psi}{du^2}.
\end{equation}
Hence
\begin{equation}
\frac{1}{\xi^{2}}\frac{d}{d\xi}\!\Big(\xi^{2}\frac{d}{d\xi}\Psi(A\xi)\Big)
= \frac{1}{\xi^{2}}\Big(2\xi A\frac{d\Psi}{du}+\xi^{2}A^2\frac{d^2\Psi}{du^2}\Big)
= A^2\frac{1}{u^{2}}\frac{d}{du}\!\Big(u^{2}\frac{d\Psi}{du}\Big).
\end{equation}
Since \(\Psi(u)\) satisfies the Emden–Chandrasekhar equation,
\begin{equation}
\frac{1}{u^{2}}\frac{d}{du}\!\Big(u^{2}\frac{d\Psi}{du}\Big)=e^{-\Psi(u)},
\end{equation}
it follows that
\begin{equation}
\frac{1}{\xi^{2}}\frac{d}{d\xi}\!\Big(\xi^{2}\frac{d}{d\xi}\Psi(A\xi)\Big)=A^{2}e^{-\Psi(A\xi)}.
\end{equation}
Now compute the right-hand side for \(\widetilde{\Psi}\):
\begin{equation}
e^{-\widetilde{\Psi(\xi)}}=e^{-[\Psi(A\xi)-2\ln A]}=A^{2}e^{-\Psi(A\xi)}.
\end{equation}
Comparing both sides, we find
\begin{equation}
\frac{1}{\xi^{2}}\frac{d}{d\xi}\!\Big(\xi^{2}\frac{d}{d\xi}\Psi(A\xi)\Big)=e^{-\widetilde{\Psi}(\xi)}.
\end{equation}
Subtracting the constant \(-2\ln A\) does not affect the derivatives, so the same equality holds with \(\widetilde{\Psi}\) in place of \(\Psi(A\xi)\). Therefore,
\begin{equation}
\boxed{\;\frac{1}{\xi^{2}}\frac{d}{d\xi}\!\Big(\xi^{2}\frac{d\widetilde{\Psi}}{d\xi}\Big)=e^{-\widetilde{\Psi}(\xi)}\;.}
\end{equation}

\subsection{Reduction of the Emden--Chandrasekhar equation to a first--order system}
The dimensionless Emden--Chandrasekhar equation for an isothermal sphere is
\begin{equation}
\frac{1}{\xi^{2}}\frac{d}{d\xi}\!\Big(\xi^{2}\frac{d\Psi}{d\xi}\Big)=e^{-\Psi(\xi)}.
\end{equation}
Define
\begin{equation}
v(\xi) := \xi^{2} \frac{d\Psi}{d\xi}.
\end{equation}
Then the equation reduces to the first-order system
\begin{equation}
\frac{d\Psi}{d\xi} = \frac{v}{\xi^{2}}, \qquad \frac{dv}{d\xi} = \xi^{2} e^{-\Psi}.
\end{equation}
Introduce scale-invariant homology variables
\begin{equation}
u(\xi) := \xi \frac{d\Psi}{d\xi} = \frac{v}{\xi}, \qquad
w(\xi) := \xi^{2} e^{-\Psi}.
\end{equation}
Then the system becomes autonomous in $(u,w)$:
\begin{equation}
\frac{du}{d\ln\xi} = -u + w, \qquad
\frac{dw}{d\ln\xi} = 2w - u w.
\end{equation}
A reduction to the phase plane is now possible. Eliminating $\ln\xi$ gives a single first-order ODE:
\begin{equation}
\frac{dw}{du} = \frac{w(2-u)}{w-u}.
\end{equation}
Once a solution curve $w(u)$ is found, one recovers $\xi$ via
\begin{equation}
d\ln\xi = \frac{du}{-u + w(u)}, \qquad \Psi' = \frac{u}{\xi}.
\end{equation}
This reduction, due to Chandrasekhar \textbf{[26]} trades the radial dependence for an autonomous flow in the homology phase plane. The definitions of $u$ and $w$ remove the explicit $\xi$-dependence and reflect the homology scaling. Phase-plane analysis of $(u,w)$ allows study of solution families, critical points, and stability properties of isothermal spheres without directly integrating the second-order ODE. Details are in \textbf{[26]}.

\section{Application to understanding the formation of red giant stars}
The isothermal gas sphere equation has a key application in astrophysics--the formation of red giant stars from progenitor Main-Sequence stars which have exhausted their thermonuclear fuel. With the exhaustion of the hydrogen fuel in its central core, a main-sequence star like the Sun evolves a 3-component structure: a helium-rich isothermal core, the product of billions of years of PP cycle hydrogen burning fusion reactions; a thin outer hydrogen-burning shell; an outer envelope of gas retaining most of the original chemical composition of the star. It was Gamow in 1938 who first suggested that the central helium core will become isothermal or at least close to isothermal \textbf{[57]}. Since $dL(r)/dr \sim \epsilon(r)$ for the luminosity and energy density, then for $L(0)\sim 0$ and $ \epsilon(r)\sim 0$ it follows that $L(r\le r_{c})\sim 0$. The temperature gradient is one of the basic equations of stellar structure [refs].

\begin{equation}
\frac{dT(r)}{dr}\sim L(r)
\end{equation}
so that $dT/dr\to 0$ for $L\sim 0$. Hence T is (very nearly) constant or isothermal in this region.

In the original Schönberg--Chandrasekhar (1942) paper \textbf{[30]} the core of an evolving star was modeled as \emph{effectively isothermal}, corresponding to a polytrope with $\gamma=1$ ($n\to\infty$), while the envelope was modeled as a polytrope with $n=3$ ($\gamma=4/3$). This formulation allowed them to derive the \emph{Schönberg--Chandrasekhar limit}, i.e., the maximum fractional mass of an isothermal helium core that can be supported by a non-isothermal envelope, without explicitly solving the full Emden--Chandrasekhar equation. An isothermal core develops only if the star is in the mass range $1.5M_{\odot}\le M\le 6 M_{\odot}$ with degeneracy contributing to the pressure and maintenance of hydrostatic equilibrium. Above this mass, contraction is much greater with a significant release of gravitational energy and an isothermal core never forms.

Consider a star that has exhausted its hydrogen fuel in its core. Let the total mass of the star satisfy $M_o \le M \le 6\,M_\odot $and let the helium-rich core initially be approximately isothermal, embedded within a polytropic envelope. Define the core-to-star mass ratio
\begin{equation}
\chi = \frac{\mathcal{M}}{M},~~\chi_{*}=\frac{\mathcal{M}_{*}}{M}
\end{equation}
where $\mathcal{M}_*$ is the maximum core mass that can remain isothermal and in hydrostatic equilibrium within the envelope. If the actual core mass exceeds this fraction, $\chi > \chi_*$, the isothermal core can no longer support itself against gravity. As a result, the core contracts and heats, inducing a temperature gradient. Once the temperature is sufficiently high, triple-alpha helium fusion can ignite, overcoming the "Berylium bottleneck" of thermonuclear reactions. This defines the \emph{Schoenberg-Chandrasekhar limit} for the core mass fraction. The isothermal core satisfies hydrostatic equilibrium
\begin{equation}
\frac{dP}{dr} = - \frac{G M(r) \rho(r)}{r^2}, \quad P = \rho \frac{k_B T}{\mu m_H},
\end{equation}
while the surrounding envelope can be modeled as a polytrope:
\begin{equation}
P = K \rho^{1 + 1/n}, \quad n \sim 3 \text{ for a radiative envelope}.
\end{equation}
The Chandrasekhar form of the Poisson-Boltzmann equation for a spherically symmetric isothermal gas sphere is
\begin{equation}
\frac{1}{r^2}\frac{d}{dr}\left( r^2 \frac{d\Psi}{dr} \right) = 4 \pi G \rho_c e^{-\Psi(r)},
\end{equation}
where
\begin{equation}
\Psi(r) = \beta m \bigl( \Phi(r) - \Phi(0) \bigr), \quad \beta = \frac{1}{k_B T}, \quad \rho_c = \rho(0).
\end{equation}
By matching the isothermal core and polytropic envelope at the interface, continuity of pressure and density leads to the SC limit:
\begin{equation}
\chi_* = \frac{\mathcal{M}_*}{M} \sim 0.37 \quad \text{for an } n=3 \text{ envelope}.
\end{equation}
Thus, the SC limit provides a quantitative criterion for the onset of core contraction and helium ignition in post-main-sequence stellar evolution.
In this core, the radial temperature gradient is nearly vanishing.

The core-envelope structure is then characterized by different polytropic indices:
\begin{equation}
n_{\rm core} \to \infty \quad (\gamma_{\rm core}=1), \qquad n_{\rm env} = 3 \quad (\gamma_{\rm env}=4/3).
\end{equation}
Later authors \textbf{[36]} applied the full isothermal sphere solution to model the quasi-isothermal cores of red giants more quantitatively. These studies \textbf{[30-36]}emphasize that, although the SC limit captures the essential stability physics, solving the Emden--Chandrasekhar equation allows detailed predictions for the core density and pressure profiles, core radius, and the approach to hydrostatic equilibrium.

\subsection{Detailed Derivation of the Schoenberg--Chandrasekhar Limit}
The Emden--Chandrasekhar Equation can be applied in more detail to see how an isothermal, non-burning stellar core of helium forms within a main-sequence star which is exhausting its hydrogen thermonuclear fuel. Consider an isothermal non-burning stellar core with
\begin{equation}
P = a^2 \rho, \qquad a^2 = \frac{k_B T}{\mu_c m_p}.
\end{equation}
Define dimensionless variables:
\begin{equation}
\Phi(\xi) \equiv -\ln\frac{\rho(r)}{\rho_c},
\qquad
\xi = \frac{r}{r_0},
\qquad
r_0 = \frac{a}{\sqrt{4\pi G \rho_c}}.
\end{equation}
Then hydrostatic equilibrium yields the Emden--Chandrasekhar ODE:
\begin{equation}
\frac{1}{\xi^2} \frac{d}{d\xi}\left(\xi^2 \frac{d\Phi}{d\xi}\right) = e^{-\Phi},
\qquad
\Phi(0) = 0, \quad \Phi'(0) = 0.
\end{equation}
The core mass inside a radius $r_b = \xi_b r_0$ is
\begin{equation}
\mathcal{M} = 4\pi \int_0^{r_b} \rho(r) r^2 dr
= 4 \pi \rho_c r_0^3 \int_0^{\xi_b} e^{-\Phi(\xi)} \xi^2 d\xi
\equiv \frac{a^3}{G^{3/2} (4\pi)^{1/2} \rho_c^{1/2}} \mathrm{J}(\xi_b),
\end{equation}
where
\begin{equation}
\mathrm{J}(\xi_b) \equiv \int_0^{\xi_b} e^{-\Phi(\xi)} \xi^2 d\xi.
\end{equation}
\begin{enumerate}
\item\underline{Polytropic Envelope}:
Assume the envelope obeys a polytropic equation of state:
\begin{equation}
P = K \rho^{1 + 1/n}.
\end{equation}
For a polytropic star of index $n$, the total mass is
\begin{equation}
M = 4 \pi \omega_n \left( \frac{K}{G} \right)^{3/2},
\qquad
\omega_n \equiv -\xi_1^2 \theta'(\xi_1),
\end{equation}
where $\theta(\xi)$ solves the Lane--Emden equation
\begin{equation}
\frac{1}{\xi^2} \frac{d}{d\xi} \left( \xi^2 \frac{d\theta}{d\xi} \right) = -\theta^n,
\quad
\theta(0)=1, \quad \theta'(0)=0,
\end{equation}
and $\xi_1$ is the first zero of $\theta(\xi)$.

\item\underline{Matching at the Core--Envelope Boundary}:
At the boundary $r_b$ between core and envelope, the pressures must match:
\begin{equation}
P_{\rm core}(r_b) = P_{\rm env}(r_b)
\quad \Rightarrow \quad
a^2 \rho_b = K \rho_b^{1 + 1/n}.
\end{equation}
Solve for $K$ in terms of core quantities:
\begin{equation}
K = a^2 \rho_b^{-1/n}.
\end{equation}
Substitute $K$ into the polytrope mass formula:
\begin{equation}
M = 4 \pi \omega_n \frac{a^3}{G^{3/2}} \rho_b^{-1/2}.
\end{equation}
Use the relation $\rho_b = \rho_c e^{-\Phi(\xi_b)}$ from the isothermal solution:
\begin{equation}
\boxed{
M = 4 \pi \omega_n \frac{a^3}{G^{3/2}} \rho_c^{-1/2} e^{\Phi(\xi_b)/2}.
}
\end{equation}

\item\underline{Core Mass Fraction}:

Divide the core mass by the total mass:
\begin{equation}
\frac{\mathcal{M}}{M} = \frac{\mathrm{J}(\xi_b) e^{\Phi(\xi_b)/2}}{(4 \pi)^{3/2} \omega_n}.
\end{equation}

This is the exact expression for the core mass fraction as a function of the isothermal Emden solution and the Lane--Emden constants.

\item\underline{Canonical Value $n=3$ Radiative Envelope}:
For a radiative envelope with $n=3$:
\begin{equation}
\xi_1 \simeq 6.89685, \qquad \omega_3 \simeq 2.01824.
\end{equation}
Numerically solving the isothermal Emden equation and maximizing
$\frac{M_c}{M}(\xi_b)$ over admissible $\xi_b$ gives the Schoenberg--Chandrasekhar limit:
\begin{equation}
\boxed{
\mathcal{F}_{\rm SC} \equiv \left( \frac{M_c}{M} \right)_{\rm max} \simeq 0.08 - 0.12 \approx 0.10.
}
\end{equation}
\end{enumerate}
\section*{Red giant formation}: When the helium core exceeds $\sim 10\%$ of the total stellar mass, the configuration becomes unstable: the core contracts under gravity while the envelope expands, producing a red giant. In low-mass stars ($M \sim 1~M_\odot$), the core becomes electron-degenerate as it contracts. Once temperatures reach $T \gtrsim 10^8~\mathrm{K}$, helium ignites via the triple-alpha process,
\begin{equation}
\mathrm{^4He + ^4He \;\longrightarrow\; ^8Be}, \qquad
^8\mathrm{Be + ^4He \;\longrightarrow\; ^{12}C + \gamma},
\end{equation}
releasing $\sim 7.275~\mathrm{MeV}$ per carbon nucleus. In degenerate cores, the ignition is thermally runaway, producing a brief but intense \emph{helium flash}, which lifts the degeneracy and establishes stable core helium burning. Meanwhile, hydrogen fusion continues in a shell surrounding the helium core, providing most of the star's luminosity during the red giant phase. The outer envelope responds to the increased core luminosity by expanding and cooling, causing the star to swell and its surface temperature to drop, giving the characteristic large radius and reddish color of a red giant (typical radii $\sim 50$–100~$R_\odot$ for Sun-like progenitors). Well-known examples of red giants with Sun-like progenitors include \emph{Arcturus} (Alpha Boötis) and \emph{Aldebaran} (Alpha Tauri).

In contrast, much more massive stars ($M \gtrsim 8~M_\odot$) do not develop strongly degenerate or nearly isothermal helium cores as gravitational compression is much greater. Helium ignition occurs gradually, without a helium flash, and the cores are typically convective, leading to successive stages of nuclear burning. These stars become red supergiants rather than ordinary red giants, with enormous radii ($\sim 500$–1000~$R_\odot$); well-known examples include \emph{Betelgeuse} (Alpha Orionis) and \emph{Antares} (Alpha Scorpii), which will eventually end their lives as core-collapse supernovae. Low-mass red giants, after exhausting their core helium and shedding their outer envelopes, leave behind collapsed carbon-oxygen 'cores' or white dwarfs supported by electron degeneracy pressure. This will be the fate of the Sun.

\section{Isothermal self-gravitating Newtonian gases and the Vlasov-Poisson formalism}
In this section we derive the classical Emden-Chandrasekhar equation starting from first principles beginning with the stationary Vlasov-Poisson equation for a self-gravitating Newtonian gas in a support domain $\Omega$.

\begin{defn}[\textbf{Phase space and distribution function}]
Let $\Omega\subset\mathbb{R}^3$ denote the spatial domain. The single-particle phase space of the gas of $N$ particles of mass $m$ is $\Gamma = \Omega \times \mathbb{R}^3 $ with points $\mathbf{Z}=(\mathbf{x},\mathbf{p})$. The phase space integral measure is $d\mathbf{Z}=d^{3}xd^{3}p$. An admissible distribution function is a nonnegative measurable function $\mathcal{F} : \Gamma \to \mathbb{R}^+ $,
satisfying $\mathcal{F}(\mathbf{x},\mathbf{p})=0$ for $\mathbf{x}\in\partial\Omega$. The total mass  is
\begin{equation}
M = \int_{\Gamma} m \mathcal{F}(\mathbf{x},\mathbf{p}) \, d^3x\, d^3p \equiv \int_{\Omega}\int_{\mathbb{R}^{3}} m \mathcal{F}(\mathbf{x},\mathbf{p}) \, d^3x\, d^3p\equiv
\int_{\Gamma}m \mathcal{F}(\mathbf{Z})\,d\mathbf{Z}< \infty,
\end{equation}
and $M=mN$. The particle number $N$, density $\rho$ and pressure $P$ integrals are
\begin{align}
&N=\int_{\Gamma}\mathcal{F}(\mathbf{x},\mathbf{p})\,d^{3}x\,d^{3}p,\\
&\rho(\mathbf{x})=\int_{\mathbb{R}^{3}}m \mathcal{F}(\mathbf{x},\mathbf{p})\,d^{3}p,\\
&P(\mathbf{x})=\frac{1}{3}\int_{\mathbb{R}^{3}}\frac{|\mathbf{p}|^{2}}{m} \mathcal{F}(\mathbf{x},\mathbf{p})\,d^{3}p,
\end{align}
with the factor of $\tfrac{1}{3}$ for an isotropic gas.
\end{defn}

\begin{defn}
We consider the Vlasov-Poisson system on a bounded domain $\Omega \subset \mathbb{R}^3$ in space, with phase-space distribution $\mathcal{F} \in C^1(\Omega \times \mathbb{R}^3)$ and gravitational potential $\Phi \in C^2(\overline{\Omega})$, assuming $\mathcal{F}$ vanishes at the boundary of $\Omega$. The \emph{non-relativistic Vlasov equation} (collisionless Boltzmann equation) for the distribution function $ \mathcal{F}(\mathbf{Z})=\mathcal{F}(\mathbf{x},\bm{p}, t) $ in phase space is
\begin{equation}
\frac{\partial \mathcal{F}}{\partial t} + \frac{\bm{p}}{m} \cdot \nabla^{(x)} \mathcal{F} - m \, \nabla^{(x)} \Phi({x},t) \cdot \nabla^{(p)} \mathcal{F} = 0,~~\mathbf{x}\in\Omega,\mathbf{p}\in\mathbb{R}^{3},
\end{equation}
and $\mathcal{F}=0$ for $\mathbf{x}\in\partial\Omega$. The Newtonian potential satisfies the usual Poisson equation
\begin{equation}
\Delta \Phi(\mathbf{x},t) = 4 \pi G \rho(\mathbf{x},t).
\end{equation}
The full component form is
\begin{equation}
\frac{\partial \mathcal{F}(\mathbf{x},\mathbf{p},t)}{\partial t}+\sum_i \frac{p_i}{m}\nabla_i^{(x)} \mathcal{F}(\mathbf{x},\mathbf{p},t) - \sum_i m \, \nabla_i^{(x)}\Phi(\mathbf{x},t)\nabla_i^{(p)}\mathcal{F}(\mathbf{x},\mathbf{p},t) = 0.
\end{equation}
\end{defn}

The Vlasov-Poisson equation can be derived from the fundamental Liouville equation, applying Hamilton's equations and the BBGKY hierarchy \textbf{[15-18]}. The Vlasov-Poisson equation
\begin{equation}
\frac{\partial \mathcal{F}}{\partial t} + \frac{\mathbf{p}}{m} \cdot \nabla_x \mathcal{F} - m \, \nabla_x \Phi(\mathbf{x}) \cdot \nabla_p \mathcal{F} = 0,
\end{equation}
with the Newtonian potential $\Phi$ satisfying
\begin{equation}
\Delta \Phi(\mathbf{x}) = 4 \pi G \rho(\mathbf{x}), \quad \rho(\mathbf{x}) = m \int \mathcal{F}(\mathbf{x},\mathbf{p}) \, d^3p,
\end{equation}
can also be interpreted as describing a \emph{single particle moving under the mean gravitational field} produced by the remaining $N-1$ particles.
In the limit $N \to \infty$, individual particle fluctuations become negligible, and the particle effectively interacts only with the smooth, self-consistent potential $\Phi(\mathbf{x})$. The mean-field formulation of the Vlasov-Poisson equation is powerful because it reduces the high-dimensional $N$-body problem to the evolution of a single particle in a self-consistent gravitational potential. This approach captures the collective dynamics and allows rigorous analysis of phenomena such as gravitational collapse, stability, and phase transitions in self-gravitating systems.

\begin{rem}
For a self-gravitating system of $N$ particles of mass $m$ interacting via Newtonian gravity, the canonical partition function
\begin{equation}
Z_N = \frac{1}{h^{3N}N!} \int_{\Gamma^N} \exp\Big[-\beta H_N(\mathbf{x}_1,\dots,\mathbf{x}_N;\mathbf{p}_1,\dots,\mathbf{p}_N)\Big]\, d^{3N}x\, d^{3N}p
\end{equation}
diverges due to the unbounded nature of the gravitational potential which has the Hamiltonian
\begin{equation}
H = \sum_{i=1}^N \frac{|\mathbf{p}_i|^2}{2m} - G \sum_{1\le i<j\le N} \frac{m^2}{|\mathbf{x}_i - \mathbf{x}_j|}.
\end{equation}
This divergence arises from configurations where two or more particles approach arbitrarily close distances, making the potential energy arbitrarily negative.
As a consequence, a canonical ensemble description fails for self-gravitating systems \textbf{[46-49]}. Our use of the Vlasov--Poisson framework circumvents this issue by considering the \emph{mean-field, collisionless limit}, where the stationary distribution $\mathcal{F}$ is well-defined and normalized over phase space:
\begin{equation}
\int_{\Gamma} \mathcal{F}_{E}(\mathbf{x},\mathbf{p})\, d^3x\,d^3p = N.
\end{equation}
Although the canonical ensemble is ill-defined for a self-gravitating gas due to divergence of the partition function, the Vlasov-Poisson equation remains valid because it describes the collisionless, mean-field evolution of the system. Stationary solutions depend on integrals of motion (Jeans theorem) and do not require canonical equilibrium.
\end{rem}

If the system is in a stationary state, then $\frac{\partial \mathcal{F}}{\partial t} = 0$, and the Vlasov equation reduces to
\begin{equation}
\frac{\bm{p}}{m} \cdot \nabla^{(x)} \mathcal{F} - m \, \nabla^{(x)} \Phi(\mathbf{x}) \cdot \nabla^{(p)} \mathcal{F} = 0.
\end{equation}
Stationary solutions will be denoted $\mathcal{F}_{E}(\mathbf{x},\mathbf{p})$ so that
\begin{equation}
\frac{\bm{p}}{m} \cdot \nabla^{(x)} \mathcal{F}_{E} - m \, \nabla^{(x)} \Phi(\mathbf{x}) \cdot \nabla^{(p)} \mathcal{F}_{E} = 0.
\end{equation}
This is the \emph{stationary Vlasov equation}, which states that the distribution function is constant along phase-space trajectories defined by the Hamiltonian system
\begin{align}
\dot{\mathbf{x}} &= \frac{\bm{p}}{m}, \\
\dot{\bm{p}} &= - m \nabla^{(x)} \Phi(\mathbf{x}).
\end{align}
Equivalently, $\frac{d\mathcal{F}_{E}}{dt} = 0$, so $\mathcal{F}_{E}$ must depend only on the integrals of motion. This is stated formally as Jeans theorem.

\begin{thm}[Jeans Theorem]
Any steady-state (time-independent) solution of the collisionless Boltzmann equation can be expressed as a function of the integrals of motion of the system. That is, if $I_1, I_2, \dots$ are integrals of motion, then
\begin{equation}
\mathcal{F}_{E}(\mathbf{x},\bm{p}) = F(I_1, I_2, \dots).
\end{equation}
\end{thm}

In any potential, the \emph{energy per particle} is conserved:
\begin{equation}
E = \frac{|\bm{p}|^2}{2m} + m \Phi(\mathbf{x}).
\end{equation}
Thus, one class of stationary solutions is
\begin{equation}
\mathcal{F}(\mathbf{x},\bm{p}) = F(E).
\end{equation}

If the potential is spherically symmetric, the \emph{angular momentum vector} $ \mathbf{L} = \mathbf{x} \wedge \bm{p}$ is conserved. In this case
\begin{equation}
\mathcal{F}(\mathbf{x},\bm{p}) = \mathcal{F}(E, |\mathbf{L}|).
\end{equation}
In an axisymmetric potential, the component $L_z$ of the angular momentum is conserved, so
\begin{equation}
\mathcal{F}(\mathbf{x},\bm{p}) = \mathcal{F}(E, L_z).
\end{equation}
Jeans theorem implies that stationary distribution functions $\mathcal{F}_{E}$ cannot be chosen arbitrarily, but must be built from conserved quantities. This principle underlies the construction of stellar equilibrium models such as the isothermal sphere and polytropes. For instance, the Maxwell--Boltzmann distribution,
\begin{equation}
\mathcal{F}(\mathbf{x},\bm{p}) \propto \exp(-\beta E),
\end{equation}
is of the form $\mathcal{F}(E)$, fully consistent with Jeans theorem. It should also be emphasized that Jeans theorem does not guarantee \emph{stability} of an equilibrium--the theorem only characterizes the form of equilibria but not their dynamical robustness.

From Jeans theorem one can guess the correct ansatz for the equilibrium solution and quickly check it. The same result is then derived in detail by maximising the Boltzmann entropy.

\begin{lem}[Stationary Maxwell--Boltzmann Distribution]
Let $\Phi(\mathbf{x})$ be a time-independent gravitational potential, and define the
single-particle energy
\begin{equation}
E(\mathbf{x},\mathbf{p}) = \frac{|\mathbf{p}|^2}{2m} + m \Phi(\mathbf{x}).
\end{equation}
Then the ansatz for Maxwell--Boltzmann distribution
\begin{equation}
\mathcal{F}_{E}(\mathbf{x},\mathbf{p}) = A \exp\big[-\beta E(\mathbf{x},\mathbf{p})\big]
\end{equation}
is a stationary solution of the Vlasov--Poisson equation:
\begin{equation}
\frac{\mathbf{p}}{m} \cdot \nabla_{\mathbf{x}} f - m\nabla_{\mathbf{x}} \Phi \cdot \nabla_{\mathbf{p}} f = 0.
\end{equation}
\begin{proof}
Compute the derivatives of $\mathcal{F}_{E}$:
\begin{equation}
\nabla_{\mathbf{x}} \mathcal{F}_{E} = -\beta \mathcal{F}_{E} \, m \nabla_{\mathbf{x}} \Phi, \qquad
\nabla_{\mathbf{p}} \mathcal{F}_{E} = -\beta \mathcal{F}_{E} \, \frac{\mathbf{p}}{m}.
\end{equation}
Substituting into the stationary Vlasov equation gives
\begin{equation}
\frac{\mathbf{p}}{m} \cdot \nabla_{\mathbf{x}} \mathcal{F}_{E} - m\nabla_{\mathbf{x}} \Phi \cdot \nabla_{\mathbf{p}} \mathcal{F}_{E}
=\frac{\mathbf{p}}{m} \cdot (-\beta \mathcal{F}_{E} m \nabla \Phi) - \nabla \Phi \cdot (-\beta \mathcal{F}_{E} m
\frac{\mathbf{p}}{m}) = 0.
\end{equation}
Thus, $\mathcal{F}_{E}$ is a stationary solution.
\end{proof}
\end{lem}

\begin{thm}[\textbf{Maxwell--Boltzmann solution of the stationary Vlasov--Poisson equation}]
Given the stationary Vlasov--Poisson equation
\begin{align}
\frac{\mathbf{p}}{m}\nabla_{x}\mathcal{F}_{E}-\nabla_{x}\Phi\nabla_{p}\mathcal{F}_{E}=0,~~\mathbf{x}\in\Omega,~\mathbf{p}\in\mathbb{R}^{3}
\end{align}
with $\mathcal{F}_{E}=0$ on $\partial\Omega$, there exists a stationary solution of the Maxwell--Boltzmann form
\begin{align}
\mathcal{F}_{E}(\mathbf{x},\mathbf{p})=A\exp\left(-\beta\left(\frac{|\mathbf{p}|^{2}}{2m}+m\Phi(x)\right)\right),
\end{align}
where $A$ is a normalisation constant. This is derived formally and rigorously by taking the variation of the Boltzmann entropy subject to the constraints of fixed particle number and fixed total energy.
\end{thm}

\begin{proof}
Let $d\mathbf{Z}=d^{3}x\,d^{3}p$. We consider the Boltzmann entropy functional for a stationary distribution $\mathcal{F}_{E}(\mathbf{x},\mathbf{p})$ over phase space $\Gamma = \Omega \times \mathbb{R}^3$:
\[
\mathcal{H}[\mathcal{F}_{E}] = - k_B \int_\Gamma \mathcal{F}_{E}(\mathbf{x},\mathbf{p}) \log \mathcal{F}_{E}(\mathbf{x},\mathbf{p}) \, d\mathbf{Z}.
\]
We wish to maximize $\mathcal{H}[\mathcal{F}_{E}]$ subject to the constraints of fixed particle number and total energy:
\[
N[\mathcal{F}_{E}] = \int_\Gamma \mathcal{F}_{E}(\mathbf{x},\mathbf{p}) \, d\mathbf{Z} = N_0,
\]
\[
E[\mathcal{F}_{E}] = \int_\Gamma \mathcal{F}_{E}(\mathbf{x},\mathbf{p}) \left( \frac{|\mathbf{p}|^2}{2m} + m \Phi(\mathbf{x}) \right) d\mathbf{Z} = E_0.
\]

Introduce Lagrange multipliers $\alpha$ and $\beta$ for these constraints, and define the constrained functional or Massieu functional
\[
\mathcal{J}[\mathcal{F}_{E}] = \mathcal{H}[\mathcal{F}_{E}] - \alpha \big(N[\mathcal{F}_{E}] - N_0\big) - \beta \big(E[\mathcal{F}_{E}] - E_0\big),
\]
which explicitly reads
\begin{align}
\mathcal{J} [\mathcal{F}_{E}]& = - k_B \int_\Gamma \mathcal{F}_{E} \log \mathcal{F}_{E} \, d\mathbf{Z}
- \alpha \left( \int_\Gamma \mathcal{F}_{E} \, d\mathbf{Z} - N_0 \right)\nonumber\\&
- \beta \left( \int_\Gamma \mathcal{F}_{E} \left( \frac{|\mathbf{p}|^2}{2m} + m \Phi(\mathbf{x}) \right) d\mathbf{Z} - E_0 \right).
\end{align}
Consider an arbitrary infinitesimal variation of the stationary distribution function $\delta \mathcal{F}_{E}(\mathbf{x},\mathbf{p})$, so that $\mathcal{F}_{E} \to \mathcal{F}_{E} + \delta \mathcal{F}_{E}$. Then we compute $\delta \mathcal{J}[\mathcal{F}_{E}]=0$ which is
\begin{align}
\delta \mathcal{J}[\mathcal{F}_{E}]& = - k_B \delta \int_\Gamma \mathcal{F}_{E} \log \mathcal{F}_{E} \, d\mathbf{Z}
- \alpha \delta\left( \int_\Gamma \mathcal{F}_{E} \, d\mathbf{Z} - N_0 \right)\nonumber\\&
- \beta \delta \left( \int_\Gamma \mathcal{F}_{E} \left( \frac{|\mathbf{p}|^2}{2m} + m \Phi(\mathbf{x}) \right) d\mathbf{Z} - E_0 \right)\nonumber\\&
= - k_B \int_\Gamma \delta\mathcal{F}_{E} \log \mathcal{F}_{E} \, d\mathbf{Z}
- \alpha \left( \int_\Gamma \delta \mathcal{F}_{E} \, d\mathbf{Z} - N_0 \right)\nonumber\\&
- \beta \left( \int_\Gamma \delta \mathcal{F}_{E} \left( \frac{|\mathbf{p}|^2}{2m} + m \Phi(\mathbf{x}) \right) d\mathbf{Z} - E_0 \right)=0.
\end{align}
Under the variation $\log (\mathcal{F}_{E})\to\log(\mathcal{F}_{E}+\delta \mathcal{F}_{E})$ which can be expanded to first order as
\begin{align}
  \log(\mathcal{F}_{E}+\delta \mathcal{F}_{E})=\log \mathcal{F}_{E}+ \frac{\delta \mathcal{F}_{E}}{\mathcal{F}_{E}}+\mathcal{O}((\delta \mathcal{F}_{E})^{2}),
\end{align}
so that
\begin{align}
(\mathcal{F}_{E}+\delta \mathcal{F}_{E})\log(\mathcal{F}_{E}+\delta \mathcal{F}_{E})=\mathcal{F}_{E}\log \mathcal{F}_{E}+\delta \mathcal{F}_{E}+\delta \mathcal{F}_{E}\log \mathcal{F}_{E}+\mathcal{O}((\delta \mathcal{F}_{E})^{2}).
\end{align}
The first variation of the entropy is then
\begin{align}
\delta \mathcal{H}[\mathcal{F}_{E}]&=\mathcal{H}[\mathcal{F}_{E}+\delta \mathcal{F}_{E}]-\mathcal{H}[\mathcal{F}_{E}]+\mathcal{O}((\delta \mathcal{J})^{2})\nonumber\\&
=-k_{B}\int_{\Gamma}(\mathcal{F}_{E}+\delta \mathcal{F}_{E})\underbrace{\log(\mathcal{F}_{E}+\delta \mathcal{F}_{E})}d\mathbf{Z}+k_{B}\int_{\Gamma}\mathcal{F}_{E}\log \mathcal{F}_{E}d\mathbf{Z}\nonumber\\&
=-k_{B}\int_{\Gamma}(\mathcal{F}_{E}+\delta \mathcal{F}_{E})\underbrace{\bigg(\log \mathcal{F}_{E}+\frac{\delta \mathcal{F}_{E}}{\mathcal{F}_{E}}+\mathcal{O}((\delta \mathcal{F}_{E})^{2})\bigg)}d\mathbf{Z}+k_{B}\int_{\Gamma}\mathcal{F}_{E}\log \mathcal{F}_{E}d\mathbf{Z}\nonumber\\&
=-k_{B}\int_{\Gamma}\left(\mathcal{F}_{E}\log \mathcal{F}_{E}+\delta \mathcal{F}_{E}+\delta\mathcal{F}_{E}\log \mathcal{F}_{E}+\mathcal{O}((\delta \mathcal{F}_{E})^{2})\right)d\mathbf{Z}
+k_{B}\int_{\Gamma}\mathcal{F}_{E}\log \mathcal{F}_{E}d\mathbf{Z}\nonumber\\&
=-k_{B}\int_{\Gamma}\delta \mathcal{F}_{E}(\log\mathcal{F}_{E}+1)d\mathbf{Z}+\mathcal{O}((\delta \mathcal{F}_{E})^{2})=-k_{B}\int_{\Gamma}\delta \mathcal{F}_{E}(\log\mathcal{F}_{E}+1)d\mathbf{Z}.
\end{align}
Then the variation of the constrained Massieu functional is
\begin{align}
\delta \mathcal{J}[\mathcal{F}_{E}]
&= \delta \mathcal{H}[\mathcal{F}_{E}] - \alpha\,\delta N[\mathcal{F}_{E}] - \beta\, \delta E[\mathcal{F}_{E}] \nonumber\\
&= -k_{B}\!\int_{\Gamma}\! \delta \mathcal{F}_{E}(\log \mathcal{F}_{E} + 1)\, d\mathbf{Z}
 - \alpha\!\int_{\Gamma}\! \delta \mathcal{F}_{E}\, d\mathbf{Z}
 - \beta\!\int_{\Gamma}\! \delta \mathcal{F}_{E}
\!\left(\frac{|\mathbf{p}|^{2}}{2m} + m\Phi(x)\right)\! d\mathbf{Z} \nonumber\\
&= \int_{\Gamma}\! \delta \mathcal{F}_{E}\,
\underbrace{\Big[-k_{B}(\log \mathcal{F}_{E} + 1) - \alpha - \beta\!\left(\frac{|\mathbf{p}|^{2}}{2m} + m\Phi(x)\right)\!\Big]}_{=\,0}\, d\mathbf{Z} = 0,
\end{align}
and $\delta N_{*}=\delta E_{*}=0$. For arbitrary variations $\delta \mathcal{F}_{E}$ we require
\begin{align}
-k_{B}(\log \mathcal{F}_{E}+1)-\alpha-\beta\left(\frac{|\mathbf{p}|^{2}}{2m}+m\Phi(x)\right)=0,
\end{align}
so that
\begin{align}
\log \mathcal{F}_{E}=-1-\frac{\alpha}{k_{B}}-\frac{\beta}{k_{B}}\left(\frac{|\mathbf{p}|^{2}}{2m}+m\Phi(x)\right).
\end{align}
Taking the exponential
\begin{align}
\mathcal{F}_{E}(\mathbf{x},\mathbf{p})&=\exp\left(-\left(1+\frac{\alpha}{k_{B}}\right)\right)\exp\left(-\beta\left(\frac{|\mathbf{p}|^{2}}{2m}+m\Phi(x)\right)\right)
\nonumber\\&
=A\exp\left(-\beta\left(\frac{|\mathbf{p}|^{2}}{2m}+m\Phi(x)\right)\right),
\end{align}
where $A = \exp\left[-\left(1+\frac{\alpha}{k_{B}}\right)\right]$.
\end{proof}

\begin{cor}
For the stationary Maxwell--Boltzmann distribution of a self-gravitating gas,
\[
\mathcal{F}_{E}(\mathbf{x},\mathbf{p}) = A \exp\Big[-\beta\Big(\frac{|\mathbf{p}|^2}{2m} + m \Phi(\mathbf{x})\Big)\Big],
\]
the mean kinetic energy per particle satisfies
\[
\Big\langle \frac{|\mathbf{p}|^2}{2m} \Big\rangle = \frac{3}{2\beta}.
\]
\end{cor}

\begin{proof}
Consider the momentum-dependent part of the distribution:
\[
\mathcal{F}_{E,p}(\mathbf{p}) = A' \exp\Big[-\beta \frac{|\mathbf{p}|^2}{2m}\Big], \quad A' = A e^{-\beta m \Phi(\mathbf{x})}.
\]
The mean kinetic energy is
\[
\Big\langle \frac{|\mathbf{p}|^2}{2m} \Big\rangle
= \frac{\displaystyle \int_{\mathbb{R}^3} \frac{|\mathbf{p}|^2}{2m} \, e^{-\beta |\mathbf{p}|^2 / 2m} \, d^3p}{\displaystyle \int_{\mathbb{R}^3} e^{-\beta |\mathbf{p}|^2 / 2m} \, d^3p}.
\]
Switching to spherical coordinates in momentum space, $|\mathbf{p}| = p$, $d^3p = 4 \pi p^2 dp$, gives
\[
\Big\langle \frac{|\mathbf{p}|^2}{2m} \Big\rangle
= \frac{\displaystyle \int_0^\infty \frac{p^2}{2m} \, e^{-\beta p^2 / 2m} \, 4 \pi p^2 dp}{\displaystyle \int_0^\infty 4 \pi p^2 e^{-\beta p^2 / 2m} dp}
= \frac{\displaystyle 2 \pi \int_0^\infty \frac{p^4}{m} e^{-\beta p^2 / 2m} dp}{\displaystyle 4 \pi \int_0^\infty p^2 e^{-\beta p^2 / 2m} dp}
= \frac{\displaystyle \int_0^\infty p^4 e^{-a p^2} dp}{2 \displaystyle \int_0^\infty p^2 e^{-a p^2} dp},
\]
where we set $a = \beta / 2m$. Using the standard Gaussian integrals
\[
\int_0^\infty x^2 e^{-a x^2} dx = \frac{\sqrt{\pi}}{4} a^{-3/2}, \quad
\int_0^\infty x^4 e^{-a x^2} dx = \frac{3 \sqrt{\pi}}{8} a^{-5/2},
\]
we find
\[
\Big\langle \frac{|\mathbf{p}|^2}{2m} \Big\rangle
= \frac{(3 \sqrt{\pi}/8) a^{-5/2}}{2 (\sqrt{\pi}/4) a^{-3/2}}
= \frac{3}{2} \frac{1}{a}
= \frac{3}{2} \frac{1}{\beta / 2m}
= \frac{3}{2\beta}.
\]
Thus each particle contributes $\frac{3}{2\beta}$ to the mean kinetic energy, as claimed.
\end{proof}

\begin{lem}
The Lagrange multiplier $\beta$ is given by $\beta = 1 / (k_{B} T)$, where $T$ is the temperature of the gas. The self-gravitating gas is then isothermal and has a density relation of the form
\begin{align}
\rho(\mathbf{x}) &= \rho_{0} \exp(-\beta m \Phi(\mathbf{x}))
= \rho_{0} \exp\!\left(-\frac{m \Phi(\mathbf{x})}{k_{B} T}\right),
\end{align}
where $\rho_{0} = A (2\pi m / \beta)^{3/2}$. The Poisson equation then becomes
\begin{align}
\Delta \Phi(\mathbf{x}) = 4\pi G \rho(\mathbf{x})
= 4\pi G \rho_{0} \exp\!\left(-\frac{m \Phi(\mathbf{x})}{k_{B} T}\right),
\end{align}
which is the Poisson--Boltzmann equation.
\end{lem}

\begin{proof}
Given the Maxwell--Boltzmann distribution $\mathcal{F}_{E}$, integrating out the momenta gives
\begin{align}
\rho(\mathbf{x})
&= \int_{\mathbb{R}^{3}} \mathcal{F}_{E}(\mathbf{x},\mathbf{p}) \, d^{3}p \nonumber \\
&= A \exp\!\left[-\beta m \Phi(\mathbf{x})\right] \int_{\mathbb{R}^{3}} \exp\!\left(-\frac{\beta |\mathbf{p}|^{2}}{2m}\right) d^{3}p \nonumber \\
&= A \left(\frac{2\pi m}{\beta}\right)^{3/2} \exp\!\left[-\beta m \Phi(\mathbf{x})\right] \nonumber \\
&= \rho_{0} \exp\!\left[-\beta m \Phi(\mathbf{x})\right].
\end{align}
From the previous corollary, the mean kinetic energy is
\begin{align}
\langle E \rangle = \Big\langle \frac{|\mathbf{p}|^{2}}{2m} \Big\rangle = \frac{3}{2\beta}.
\end{align}
However, the equipartition theorem gives the mean kinetic energy per particle in a gas at temperature $T$ as
\[
\langle E \rangle = \frac{3}{2} k_{B} T.
\]
Equating the two expressions,
\begin{align}
\frac{3}{2\beta} = \frac{3}{2} k_{B} T,
\end{align}
hence
\begin{align}
\beta = \frac{1}{k_{B} T}.
\end{align}
The stationary Maxwell--Boltzmann distribution for the self-gravitating isothermal gas at temperature $T$ is therefore
\begin{align}
\mathcal{F}_{E}(\mathbf{x},\mathbf{p})
= A \exp\!\left[-\beta \left(\frac{|\mathbf{p}|^{2}}{2m} + m \Phi(\mathbf{x})\right)\right]
= A \exp\!\left[-\frac{1}{k_{B} T} \left(\frac{|\mathbf{p}|^{2}}{2m} + m \Phi(\mathbf{x})\right)\right].
\end{align}
\end{proof}
\subsection{Second Variation of Entropy and Relation to Antonov Instability}

Let the Boltzmann entropy of a self-gravitating collisionless gas be
\begin{equation}
\mathcal{H}[\mathcal{F}] = - k_B \int_{\Gamma} \mathcal{F}(\mathbf{x}, \mathbf{p}) \log \mathcal{F}(\mathbf{x}, \mathbf{p}) \, d^3x \, d^3p,
\end{equation}
where $\Gamma = \Omega \times \mathbb{R}^3$ is the phase-space domain.
The total mass and energy are
\begin{align}
M[\mathcal{F}] &= \int_{\Gamma} \mathcal{F}(\mathbf{x}, \mathbf{p}) \, d^3x \, d^3p = M_0, \\
E[\mathcal{F}] &= \int_{\Gamma} \mathcal{F}(\mathbf{x}, \mathbf{p}) \left( \frac{|\mathbf{p}|^2}{2m} + m \Phi(\mathbf{x}) \right) d^3x \, d^3p = E_0,
\end{align}
with $\Phi$ satisfying Poisson’s equation
\begin{equation}
\Delta \Phi = 4 \pi G m \int_{\mathbb{R}^3} \mathcal{F}(\mathbf{x}, \mathbf{p}) \, d^3p.
\end{equation}

Define the Massieu functional
\begin{equation}
\mathcal{J}_{E}[\mathcal{F}] = S[\mathcal{F}] - \alpha \big( M[\mathcal{F}] - M_0 \big) - \beta \big( E[\mathcal{F}] - E_0 \big),
\end{equation}
where $\alpha$ and $\beta$ are Lagrange multipliers.
The first variation $\delta \mathcal{J}_{E} = 0$ gives
\begin{equation}
\frac{\delta S}{\delta \mathcal{F}} - \alpha - \beta \left( \frac{|\mathbf{p}|^2}{2m} + m \Phi(\mathbf{x}) \right) = 0,
\end{equation}
leading to the Maxwell--Boltzmann distribution
\begin{equation}
\mathcal{F}_{\mathrm{MB}}(\mathbf{x}, \mathbf{p}) = A \exp\Big[-\beta \Big( \frac{|\mathbf{p}|^2}{2m} + m \Phi(\mathbf{x}) \Big) \Big],
\qquad A = e^{-\alpha - 1}.
\end{equation}

\subsubsection*{\underline{Second Variation}}

The second variation of the constrained functional is
\begin{equation}
\delta^2 \mathcal{J}_{E} = \delta^2 S - \alpha \, \delta^2 M - \beta \, \delta^2 E.
\end{equation}
Since $M[\mathcal{F}]$ and the kinetic part of $E[\mathcal{F}]$ are linear in $\mathcal{F}$, we have $\delta^2 M = 0$, and the only nonzero second-order term comes from the potential:
\begin{align}
\delta^2 \mathcal{H} &= - k_B \int_{\Gamma} \frac{(\delta \mathcal{F})^2}{\mathcal{F}_{\mathrm{MB}}} \, d^3x \, d^3p, \\
\delta^2 E &= m^2 \int_{\Omega} \delta \rho(\mathbf{x}) \, \delta \Phi(\mathbf{x}) \, d^3x,
\end{align}
where $\delta \rho(\mathbf{x}) = \int_{\mathbb{R}^3} \delta \mathcal{F}(\mathbf{x}, \mathbf{p}) \, d^3p$ and $\delta \Phi$ satisfies $\Delta \delta \Phi = 4 \pi G m \, \delta \rho$.
Thus, the complete second variation reads
\begin{equation}
\boxed{
\delta^2 \mathcal{J}_{E} = - k_B \int_{\Gamma} \frac{(\delta \mathcal{F})^2}{\mathcal{F}_{\mathrm{MB}}} \, d^3x \, d^3p
- \beta m^2 \int_{\Omega} \delta \rho(\mathbf{x}) \, \delta \Phi(\mathbf{x}) \, d^3x.
}
\end{equation}

For self-gravitating systems in a finite volume, the Maxwell--Boltzmann equilibrium may become unstable under certain conditions.
In a collisionless Vlasov--Poisson formulation, the stationary distribution
\[
\frac{\mathbf{p}}{m}\cdot\nabla_{\mathbf{x}}\mathcal{F} - m\nabla_{\mathbf{x}}\Phi \cdot \nabla_{\mathbf{p}}\mathcal{F} = 0,
\qquad \Delta \Phi = 4 \pi G \rho,
\]
has the Maxwell--Boltzmann solution
\[
\mathcal{F}_{\rm MB}(\mathbf{x},\mathbf{p})
= A \exp\Big[-\beta\Big(\frac{|\mathbf{p}|^2}{2m} + m \Phi(\mathbf{x})\Big)\Big],
\]
which self-consistently leads to the Poisson--Boltzmann equation and the classical isothermal sphere.

While the Antonov instability \textbf{[37,38]}concerns the \emph{thermodynamic stability} of self-gravitating systems in the microcanonical ensemble, where the total energy \(E\) and particle number \(N\) are fixed, our analysis remains within the \emph{collisionless mean-field (Vlasov)} framework.
The equilibrium derived here is dynamically stationary, and the canonical--microcanonical distinctions that underlie the Antonov instability do not arise at the level of the first variation.

Reduction to a density quadratic form and Sturm--Liouville problem proceeds as follows.
Consider a small perturbation of the equilibrium distribution:
\[
\mathcal{F} = \mathcal{F}_{\rm MB} + \delta \mathcal{F}, \qquad |\delta \mathcal{F}| \ll \mathcal{F}_{\rm MB}.
\]
The second variation of the Boltzmann (or Massieu) functional can be written as
\[
\delta^2 \mathcal{J}_{E}[\delta \mathcal{F}]
= - k_B \int_\Gamma \frac{(\delta \mathcal{F})^2}{\mathcal{F}_{\rm MB}} \, d^3x\,d^3p
- \frac{\beta}{2} \int_\Omega \delta \rho \, \delta \Phi \, d^3x,
\]
where
\[
\delta \rho(\mathbf{x}) = m \int_{\mathbb{R}^3} \delta \mathcal{F}(\mathbf{x},\mathbf{p}) \, d^3p,
\qquad \Delta \delta \Phi = 4 \pi G \, \delta \rho.
\]

For a given density perturbation \(\delta \rho(\mathbf{x})\), the phase-space perturbation \(\delta \mathcal{F}\) that minimizes \(\delta^2 \mathcal{J}_{E}\) is
\[
\delta \widehat{\mathcal{F}}(\mathbf{x},\mathbf{p})
= \frac{\delta \rho(\mathbf{x})}{\rho(\mathbf{x})} \, \mathcal{F}_{\rm MB}(\mathbf{x},\mathbf{p}),
\]
which follows from integrating out momentum degrees of freedom.
Substituting $\delta \widehat{\mathcal{F}}$ back, the second variation reduces to a \emph{density-only quadratic form}:
\[
\delta^2 \mathcal{J}_{E}[\delta \rho]
= - k_B \int_\Omega \frac{(\delta \rho)^2}{\rho} \, d^3x
- \frac{\beta}{2} \int_\Omega \delta \rho \, \delta \Phi \, d^3x.
\]

Assuming spherical symmetry and defining \(\mathcal{U}(r) = r^2 \delta \Phi(r)\), this quadratic form leads to the Sturm--Liouville eigenvalue problem
\[
\frac{d}{dr} \left( \frac{r^2}{\rho(r)} \frac{d\mathcal{U}}{dr} \right)
+ 4 \pi G \beta \, r^2 \mathcal{U} = \lambda \mathcal{U},
\qquad \mathcal{U}(0) \text{ finite}, \quad \mathcal{U}(R) = 0,
\]
where the lowest eigenvalue crossing zero identifies the onset of the Antonov instability, corresponding to the critical energy
\[
E_c \simeq -0.335 \, \frac{G M^2}{R}.
\]
This formulation provides a rigorous reduction from phase-space perturbations to a density-only eigenvalue problem, making the connection between the Maxwell--Boltzmann equilibrium, dynamical stationarity, and the classical Antonov instability transparent.

\begin{thm}
Let $\mathcal{F}_E(\mathbf{x},\mathbf{p})$ be a stationary Maxwell--Boltzmann phase--space mass density,
\[
\mathcal{F}_E(\mathbf{x},\mathbf{p}) = A \exp\!\left[-\beta\!\left(\frac{|\mathbf{p}|^2}{2m} + m\Phi(\mathbf{x})\right)\right],
\qquad \beta=\frac{1}{k_B T},
\]
with $A>0$ constant.
Assume spherical symmetry $\Phi=\Phi(r)$ and define the spatial mass density
\[
\rho(\mathbf{x}) = \int_{\mathbb{R}^3} \mathcal{F}_E(\mathbf{x},\mathbf{p})\,d^3p.
\]
Then, upon substitution into Poisson's equation $\Delta\Phi = 4\pi G \rho$, and after the standard rescaling
\[
\Psi(r) := \beta m\big(\Phi(r)-\Phi(0)\big), \qquad
\xi := \frac{r}{r_0}, \qquad
r_0^2 := \frac{k_B T}{4\pi G m \rho_c},
\]
where $\rho_c:=\rho(0)$ is the central density, one obtains the dimensionless Emden--Chandrasekhar equation
\[
\boxed{
\frac{1}{\xi^2}\frac{d}{d\xi}\!\Big(\xi^2\frac{d\Psi}{d\xi}\Big) = e^{-\Psi(\xi)},
\qquad \Psi(0)=0,\quad \Psi'(0)=0.
}
\]
\end{thm}

\begin{proof}
We first compute the spatial density by integrating the Maxwell--Boltzmann distribution over momentum. Starting from
\[
\mathcal{F}_E(\mathbf{x},\mathbf{p}) = A \exp\!\Big[-\beta\Big(\frac{|\mathbf{p}|^2}{2m}+m\Phi(\mathbf{x})\Big)\Big],
\]
the momentum integral is Gaussian:
\[
\int_{\mathbb{R}^3} \exp\!\Big(-\beta\frac{|\mathbf{p}|^2}{2m}\Big)\,d^3p
= \left(\frac{2\pi m}{\beta}\right)^{3/2}.
\]
Therefore,
\[
\rho(\mathbf{x}) = \int_{\mathbb{R}^3} \mathcal{F}_E(\mathbf{x},\mathbf{p})\,d^3p
= A\left(\frac{2\pi m}{\beta}\right)^{3/2} e^{-\beta m \Phi(\mathbf{x})}.
\]
Define
\[
\rho_0 := A\left(\frac{2\pi m}{\beta}\right)^{3/2},
\]
so that
\[
\rho(\mathbf{x}) = \rho_0 \, e^{-\beta m \Phi(\mathbf{x})}.
\]

Next, introduce the central density and the dimensionless potential. Writing the density relative to the centre, let
\[
\rho_c := \rho(0) = \rho_0 e^{-\beta m \Phi(0)}.
\]
Define the dimensionless potential
\[
\Psi(r) := \beta m\big(\Phi(r)-\Phi(0)\big),
\]
so that
\[
e^{-\beta m \Phi(r)} = e^{-\beta m \Phi(0)} e^{-\Psi(r)} = \frac{\rho_c}{\rho_0}\, e^{-\Psi(r)},
\]
and therefore the spatial density has the usual isothermal form
\[
\rho(r) = \rho_c \, e^{-\Psi(r)}.
\]

The spherically symmetric Poisson equation reads
\[
\Delta\Phi(r) = \frac{1}{r^2}\frac{d}{dr}\!\Big(r^2\frac{d\Phi}{dr}\Big)
= 4\pi G \rho(r) = 4\pi G \rho_c e^{-\Psi(r)}.
\]
Introducing the dimensionless radius by the scaling
\[
r_0^2 := \frac{k_B T}{4\pi G m \rho_c},
\qquad \xi := \frac{r}{r_0},
\]
we note that \(k_B T = 1/\beta\), so that
\[
4\pi G \rho_c = \frac{1}{\beta m r_0^2}.
\]

Expressing derivatives with respect to \(r\) in terms of \(\xi\),
\[
\frac{d}{dr} = \frac{1}{r_0}\frac{d}{d\xi}, \qquad
\frac{d}{dr}\!\Big(r^2\frac{d\Phi}{dr}\Big)
= \frac{1}{r_0^2}\frac{d}{d\xi}\!\Big(\xi^2\frac{d\Phi}{d\xi}\Big).
\]
From \(\Psi = \beta m(\Phi-\Phi(0))\), we have
\[
\frac{d\Phi}{dr} = \frac{1}{\beta m r_0}\frac{d\Psi}{d\xi},
\]
and hence
\[
\frac{1}{r^2}\frac{d}{dr}\!\Big(r^2\frac{d\Phi}{dr}\Big)
= \frac{1}{\beta m r_0^2}\frac{1}{\xi^2}\frac{d}{d\xi}\!\Big(\xi^2\frac{d\Psi}{d\xi}\Big).
\]

Substituting into Poisson’s equation yields
\[
\frac{1}{\beta m r_0^2}\frac{1}{\xi^2}\frac{d}{d\xi}\!\Big(\xi^2\frac{d\Psi}{d\xi}\Big)
= 4\pi G \rho_c e^{-\Psi}.
\]
Using \(4\pi G \rho_c = 1/(\beta m r_0^2)\), all prefactors cancel, and we obtain the dimensionless equation
\[
\frac{1}{\xi^2}\frac{d}{d\xi}\!\Big(\xi^2\frac{d\Psi}{d\xi}\Big) = e^{-\Psi(\xi)}.
\]

Regularity at the origin requires \(\Psi(0)=0\) (by definition) and \(\Psi'(0)=0\) (by spherical symmetry). These are the standard central boundary conditions of the Emden--Chandrasekhar equation. Combining the dimensionless reduction and these conditions completes the proof.
\end{proof}

\begin{rem}
Purely isothermal spheres possess a density profile
\[
\rho(r) \sim \frac{1}{r^2}, \qquad r \to \infty,
\]
leading to an \emph{infinite total mass}. In realistic astrophysical systems, such as globular clusters or galactic bulges, this unphysical behaviour is avoided by truncating the Maxwell--Boltzmann distribution at a finite escape energy. The resulting \emph{King model} distribution \textbf{[58-62]} is
\[
f_{\rm K}(\mathbf{x},\mathbf{p}) =
\begin{cases}
A\!\left[\exp\!\Big(-\beta(\mathcal{E}-\mathcal{E}_{\rm max})\Big) - 1\right], & \mathcal{E} < \mathcal{E}_{\rm max}, \\[0.5em]
0, & \mathcal{E} \ge \mathcal{E}_{\rm max},
\end{cases}
\]
where $\mathcal{E} = \frac{|\mathbf{p}|^2}{2m} + m\Phi(\mathbf{x})$ is the single-particle energy, $\mathcal{E}_{\rm max}$ the cutoff (escape) energy, and $A$ a normalization constant. The truncation yields a finite total mass and a density profile that transitions from an approximately isothermal core to a rapidly decreasing halo, contrasting with the unbounded $1/r^2$ tail of the ideal isothermal sphere.
\end{rem}
\section{Hydrostatic Equilibrium from the Maxwell--Boltzmann Distribution}

For a self-gravitating gas in thermal equilibrium, isothermality implies hydrostatic equilibrium. The standard equation of hydrostatic equilibrium follows naturally from the Maxwell--Boltzmann distribution.

\begin{thm}[Hydrostatic equilibrium from the Maxwell--Boltzmann distribution]
Let a collisionless gas in thermal equilibrium at temperature $T$ have the single-particle Maxwell--Boltzmann distribution
\begin{equation}
\mathcal{F}_{\mathrm{MB}}(\mathbf{x},\mathbf{p}) = A \exp\Big[-\beta \Big(\frac{|\mathbf{p}|^2}{2m} + m \Phi(\mathbf{x})\Big)\Big], \qquad \beta = \frac{1}{k_B T},
\end{equation}
where $\Phi(\mathbf{x})$ is the gravitational potential. Define the mass density
\begin{equation}
\rho(\mathbf{x}) = m \int_{\mathbb{R}^3} \mathcal{F}_{\mathrm{MB}}(\mathbf{x},\mathbf{p}) \, d^3 p,
\end{equation}
and the scalar pressure
\begin{equation}
P(\mathbf{x}) = \frac{1}{3} \int_{\mathbb{R}^3} \frac{|\mathbf{p}|^2}{m} \, \mathcal{F}_{\mathrm{MB}}(\mathbf{x},\mathbf{p}) \, d^3 p.
\end{equation}
Then $\rho(\mathbf{x})$ and $P(\mathbf{x})$ satisfy the hydrostatic equilibrium equation
\begin{equation}
\nabla P(\mathbf{x}) + \rho(\mathbf{x}) \nabla \Phi(\mathbf{x}) = 0.
\end{equation}
\end{thm}

\begin{proof}
The mass density can be computed by integrating the Maxwell--Boltzmann distribution over momentum space. Factorizing the exponential, we have
\begin{align}
\rho(\mathbf{x}) &= m \int_{\mathbb{R}^3} \mathcal{F}_{\mathrm{MB}}(\mathbf{x},\mathbf{p}) \, d^3p \\
&= m A \, e^{-\beta m \Phi(\mathbf{x})} \int_{\mathbb{R}^3} \exp\Big[- \beta \frac{|\mathbf{p}|^2}{2m} \Big] \, d^3p.
\end{align}

The integral over momentum is a standard Gaussian integral. Writing $\alpha = \beta/(2m)$, we separate the components:
\begin{align}
\int_{\mathbb{R}^3} e^{-\alpha |\mathbf{p}|^2} d^3p &= \int_{-\infty}^{\infty} e^{-\alpha p_x^2} dp_x \int_{-\infty}^{\infty} e^{-\alpha p_y^2} dp_y \int_{-\infty}^{\infty} e^{-\alpha p_z^2} dp_z \\
&= \left( \sqrt{\frac{\pi}{\alpha}} \right)^3.
\end{align}

Substituting back, the mass density is
\begin{equation}
\rho(\mathbf{x}) = m A \left( \frac{2 \pi m}{\beta} \right)^{3/2} e^{-\beta m \Phi(\mathbf{x})} \equiv C \, e^{-\beta m \Phi(\mathbf{x})},
\end{equation}
where $C = m A (2 \pi m / \beta)^{3/2}$ is a normalization constant.

The pressure is computed similarly:
\begin{equation}
P(\mathbf{x}) = \frac{1}{3} \int_{\mathbb{R}^3} \frac{|\mathbf{p}|^2}{m} \mathcal{F}_{\mathrm{MB}}(\mathbf{x},\mathbf{p}) \, d^3 p
= \frac{A e^{-\beta m \Phi(\mathbf{x})}}{3 m} \int_{\mathbb{R}^3} |\mathbf{p}|^2 \exp(-\alpha |\mathbf{p}|^2) \, d^3p.
\end{equation}

Switching to spherical coordinates $p = |\mathbf{p}|$, $d^3p = 4 \pi p^2 dp$, the integral becomes
\begin{align}
\int_{\mathbb{R}^3} |\mathbf{p}|^2 e^{-\alpha |\mathbf{p}|^2} d^3p &= 4 \pi \int_0^\infty p^4 e^{-\alpha p^2} dp.
\end{align}

Using the Gamma function identity $\int_0^\infty x^{n} e^{-a x} dx = \Gamma(n+1) / a^{n+1}$ with $x = p^2$, $dx = 2p dp$, we have
\begin{align}
4 \pi \int_0^\infty p^4 e^{-\alpha p^2} dp &= 2 \pi \int_0^\infty p^3 e^{-\alpha p^2} \cdot 2p \, dp = 2 \pi \cdot \frac{3 \sqrt{\pi}}{4} \alpha^{-5/2} = \frac{3}{2 \alpha} \left( \frac{\pi}{\alpha} \right)^{3/2}.
\end{align}

Substituting $\alpha = \beta/(2m)$, we find
\begin{equation}
P(\mathbf{x}) = \frac{A e^{-\beta m \Phi(\mathbf{x})}}{\beta} \left( \frac{2 \pi m}{\beta} \right)^{3/2}.
\end{equation}

Comparing with the density, we recover the ideal gas relation:
\begin{equation}
P(\mathbf{x}) = \frac{1}{\beta m} \rho(\mathbf{x}) = \frac{k_B T}{m} \rho(\mathbf{x}).
\end{equation}

Since the temperature is constant throughout the gas, the gradients of the pressure and density are
\begin{align}
\nabla P(\mathbf{x}) &= \frac{1}{\beta m} \nabla \rho(\mathbf{x}) \\
\nabla \rho(\mathbf{x}) &= - \beta m \rho(\mathbf{x}) \nabla \Phi(\mathbf{x}) \\
\Rightarrow \nabla P(\mathbf{x}) &= - \rho(\mathbf{x}) \nabla \Phi(\mathbf{x}),
\end{align}
which is precisely the equation of hydrostatic equilibrium. Hence, the Maxwell--Boltzmann distribution implies that an isothermal, self-gravitating gas automatically satisfies the hydrostatic equilibrium condition.
\end{proof}

\begin{rem}
Although the Vlasov--Poisson system describes a collisionless mean-field dynamics, its stationary solutions reproduce the same macroscopic equilibrium relations as those
of a collisional, isothermal gas. By Jeans' theorem, any stationary distribution must depend only on the integrals of motion,
so that $f = \mathcal{F}_{E}(E)$ with $E = \tfrac{1}{2} m v^2 + m \Phi(\mathbf{x})$. Maximizing the Boltzmann entropy under fixed total mass and energy yields the
Maxwell--Boltzmann form
\[
\mathcal{F}_{E} = A \exp[-\beta m (\tfrac{1}{2}v^2 + \Phi)],
\]
which is automatically a stationary solution of the Vlasov equation. Integrating over velocity space gives $\rho = \rho_0 e^{-\beta m(\Phi - \Phi_0)}$,
leading to the Poisson--Boltzmann equation and, consequently, to the hydrostatic equilibrium relation $\nabla P = -\rho \nabla \Phi$ with $P = (k_B T/m)\rho$. Hence, despite being collisionless, the Vlasov--Poisson framework recovers the same equilibrium structure as that of an isothermal gas in collisional thermodynamic equilibrium. The difference lies not in the form of equilibrium but in the underlying relaxation mechanism and stability properties. Although the microscopic physics differ (collisionless mean field vs. collisional thermodynamics), the stationary distribution functions coincide because both are derived by maximizing the same entropy functional under the same constraints (fixed mass and energy). The Vlasov–Poisson system naturally admits an entropy extremum that has the same form as the Maxwell–Boltzmann distribution — even though the dynamics that would bring it there (two-body collisions) are absent. Thus, the formal coincidence arises from shared statistical principles, not from identical dynamics.
\end{rem}

\subsection*{Relation to Stellar Isothermal Cores}
The isothermal sphere derived above has a direct analogue in stellar structure theory. Prior to the onset of the red-giant phase, as discussed in Section 1, the core of an evolving star can be approximated as an isothermal, self-gravitating gas in hydrostatic equilibrium,
\[
\nabla P = -\rho \nabla \Phi, \qquad P = \frac{k_B T}{m}\rho,
\]
which leads to the same Poisson--Boltzmann equation,
\[
\Delta \Phi = 4\pi G \rho_0 e^{-\beta m(\Phi - \Phi_0)},
\]
as obtained from the stationary Vlasov--Poisson equation with a Maxwell--Boltzmann distribution. The mathematical form of the equilibrium is therefore identical in both cases.
The physical interpretation, however, is different. In a stellar core, frequent particle collisions maintain local thermodynamic equilibrium, and stability is governed by the equations of hydrodynamics and energy transport rather than by the entropy functional of a collisionless system. Nevertheless, when the core becomes too centrally condensed and its specific heat turns negative, hydrostatic balance is lost, leading to core contraction and envelope expansion—the onset of the red-giant phase. This hydrodynamic instability is formally analogous to the Antonov (gravothermal) instability of a collisionless isothermal sphere: in both cases, equilibrium ceases to exist once the density contrast exceeds a critical value or, equivalently, when the total energy drops below the critical value $E_c \simeq -0.335\, G M^2 / R$. The analogy thus provides a unifying view linking kinetic and stellar-structure descriptions of
self-gravitating thermal systems.

\section{Total Energy Integral of a Self-Gravitating System from Vlasov equation}
Start from the non-relativistic Vlasov equation,
\begin{equation}
\frac{\partial f}{\partial t} + \sum_i \frac{p_i}{m} \cdot \nabla_i^{(x)} f - \sum_i m \, \nabla_i^{(x)} \Phi(\mathbf{x}) \cdot \nabla_i^{(p)} f = 0,
\end{equation}
where $\mathcal{F}(\mathbf{x},{p})$ is the single-particle distribution function, $\Phi(\mathbf{x})$ is the gravitational potential, and $\nabla_i^{(p)} = \frac{\partial}{\partial p_i}$. Multiplying the Vlasov equation by the single-particle energy
\begin{equation}
\varepsilon(\mathbf{x},{p}) = \frac{|p|^2}{2 m} + \frac{1}{2} m \Phi(\mathbf{x})
\end{equation}
and integrating over momentum and position gives
\begin{equation}
\int d^3 x \int d^3 p \, \varepsilon(\mathbf{x},{p}) \left[ \frac{\partial f}{\partial t} + \sum_i \frac{p_i}{m} \cdot \nabla_i^{(x)} f - \sum_i m \, \nabla_i^{(x)} \Phi \cdot \nabla_i^{(p)} f \right] = 0.
\end{equation}

The kinetic energy contribution is
\begin{equation}
K = \int d^3 x \int d^3 p \, \frac{|p|^2}{2 m} \mathcal{F}(\mathbf{x},{p}),
\end{equation}
and its time derivative gives the rate of change of kinetic energy:
\begin{equation}
\frac{d K}{d t} = \int d^3 x \int d^3 p \, \frac{|p|^2}{2 m} \frac{\partial f}{\partial t}.
\end{equation}

The momentum advection term vanishes upon integration by parts and assuming $f \to 0$ at infinity:
\begin{equation}
\int d^3 x \int d^3 p \, \frac{|p|^2}{2 m} \sum_i \frac{p_i}{m} \cdot \nabla_i^{(x)} f = 0.
\end{equation}

The potential term also vanishes for a system with zero bulk velocity, leaving only the potential energy contribution:
\begin{equation}
W = \frac{1}{2} \int d^3 x \, \rho(\mathbf{x}) \Phi(\mathbf{x}), \quad \rho(\mathbf{x}) = \int m \mathcal{F}(\mathbf{x},{p}) \, d^3 p.
\end{equation}

Combining kinetic and potential terms, the total energy is
\begin{equation}
E_{tot} = K + W = \int d^3 x \int d^3 p \frac{|p|^2}{2 m} \mathcal{F}(\mathbf{x},{p}) +\frac{1}{2} \int d^3 x \rho(\mathbf{x})\Phi(\mathbf{x}),
\end{equation}
and is conserved, $\frac{d E_{tot}}{d t} = 0$.

\subsection{Astrophysical Form Using Internal Energy Density}
Defining the internal energy density as
\begin{equation}
\mathcal{E}(\mathbf{x}) = \int d^3 p \, \frac{|p|^2}{2 m} \mathcal{F}(\mathbf{x},{p}),
\end{equation}
the total energy can be written in the familiar astrophysical form:
\begin{equation}
E_{tot} = \int d^3 x \, \mathcal{E}(\mathbf{x}) + \frac{1}{2} \int d^3 x \rho(\mathbf{x}) \Phi(\mathbf{x}).
\end{equation}

Hydrostatic equilibrium then emerges as a critical point of the total energy under mass-conserving variations $\delta \rho(\mathbf{x})$:
\begin{equation}
\delta E_{tot} = \delta \int d^3 x \, \mathcal{E}(\rho) + \delta \frac{1}{2} \int d^3 x \rho \Phi = 0,
\end{equation}
which leads to the standard hydrostatic equilibrium equation:
\begin{equation}
\nabla P(\mathbf{x}) + \rho(\mathbf{x}) \nabla \Phi(\mathbf{x}) = 0,
\end{equation}
where $P(\mathbf{x})$ is the pressure, related to $\mathcal{E}(\mathbf{x})$ by $\delta \mathcal{E} = \frac{P}{\rho} \delta \rho$.
\subsection{Maxwell-Boltzmann Distribution and Explicit Energy Evaluation}

\begin{lem}[\textbf{Maxwell–Boltzmann Distribution and Energy Decomposition}]
For a self-gravitating system of identical particles of mass \(m\) in thermal equilibrium at temperature \(T\),
the single-particle distribution function is
\begin{equation}
f_{MB}(\mathbf{x}, \mathbf{p})
= A \exp\!\Bigg[-\frac{1}{k_B T} \Big(\frac{|\mathbf{p}|^2}{2m} + m \Phi(\mathbf{x})\Big)\Bigg],
\end{equation}
where \(\Phi(\mathbf{x})\) is a self-consistent gravitational potential satisfying Poisson’s equation,
and \(A\) is a normalization constant determined by the total mass constraint.

The total energy of the configuration is the sum of kinetic and potential contributions:
\begin{equation}
E_{tot} = \int d^3x \, \mathcal{E}(\mathbf{x}) + \frac{1}{2} \int d^3x \, \rho(\mathbf{x}) \Phi(\mathbf{x}),
\end{equation}
where the local internal energy density is \(\mathcal{E}(\mathbf{x}) = \tfrac{3}{2} k_B T \, \rho(\mathbf{x})/m\).
Furthermore, the equilibrium configuration satisfies hydrostatic balance,
\begin{equation}
\nabla P(\mathbf{x}) + \rho(\mathbf{x}) \nabla \Phi(\mathbf{x}) = 0,
\end{equation}
with \(P(\mathbf{x}) = \rho(\mathbf{x}) k_B T / m\).
\end{lem}

\begin{proof}
Start with the Maxwell–Boltzmann distribution
\begin{equation}
f_{MB}(\mathbf{x}, \mathbf{p})
= A \exp\!\Bigg[-\frac{1}{k_B T} \Big(\frac{|\mathbf{p}|^2}{2m} + m \Phi(\mathbf{x})\Big)\Bigg].
\end{equation}

The kinetic energy can be written as
\begin{equation}
K = \int d^3x \, e^{- m \Phi(\mathbf{x}) / k_B T}
     \int d^3p \, \frac{|\mathbf{p}|^2}{2m} A e^{- |\mathbf{p}|^2 / (2m k_B T)}.
\end{equation}

The momentum integral is a standard Gaussian:
\begin{equation}
\int d^3p \, \frac{|\mathbf{p}|^2}{2m} e^{- |\mathbf{p}|^2 / (2m k_B T)}
= \frac{3}{2} (2\pi m k_B T)^{3/2} k_B T.
\end{equation}

Thus the local internal energy density is
\begin{equation}
\mathcal{E}(\mathbf{x}) = \frac{3}{2} k_B T \, \frac{\rho(\mathbf{x})}{m},
\end{equation}
where the mass density is obtained from integrating over momenta,
\begin{equation}
\rho(\mathbf{x}) = m \int f_{MB}(\mathbf{x}, \mathbf{p}) \, d^3p
= m A (2\pi m k_B T)^{3/2} e^{- m \Phi(\mathbf{x}) / k_B T}.
\end{equation}

The gravitational potential energy is
\begin{equation}
W = \frac{1}{2} \int d^3x \, \rho(\mathbf{x}) \Phi(\mathbf{x}).
\end{equation}

Hence the total energy is
\begin{equation}
E_{tot} = K + W = \int d^3x \, \mathcal{E}(\mathbf{x})
+ \frac{1}{2} \int d^3x \, \rho(\mathbf{x}) \Phi(\mathbf{x}),
\end{equation}
which coincides with the expression obtained from the Vlasov (mean-field) formulation.

For a static configuration with no bulk motion, equilibrium corresponds to stationarity of the total energy under mass-conserving variations:
\begin{equation}
\delta E_{tot} = \delta \int d^3x \, \mathcal{E}(\rho)
+ \delta \frac{1}{2} \int d^3x \, \rho \Phi = 0.
\end{equation}

Using \(\delta \mathcal{E} = (P/\rho)\,\delta\rho\) and varying \(\Phi\) consistently through Poisson’s equation yields
\begin{equation}
\nabla P(\mathbf{x}) + \rho(\mathbf{x}) \nabla \Phi(\mathbf{x}) = 0,
\end{equation}
which is the condition of hydrostatic equilibrium.
Since \(P = \rho k_B T / m\) for a Maxwell–Boltzmann gas, this completes the proof.
\end{proof}

\begin{lem}[\textbf{Scalar virial theorem}]
Let $\mathcal{F}(\mathbf{x},\mathbf{p})$ be a stationary solution of the Vlasov--Poisson equation
\begin{equation}
\frac{\mathbf{p}}{m}\!\cdot\!\nabla_{\mathbf{x}} f - m\nabla_{\mathbf{x}}\Phi(\mathbf{x})\!\cdot\!\nabla_{\mathbf{p}} f = 0
\qquad\text{on }\Gamma=\Omega\times\mathbb{R}^3,
\end{equation}
where $\Omega\subset\mathbb{R}^3$ is the physical domain (either all space or a bounded domain). Assume the decay / boundary conditions
\begin{itemize}
  \item $\mathcal{F}(\mathbf{x},\mathbf{p})\to 0$ sufficiently fast as $|\mathbf{p}|\to\infty$ so that momentum-space surface integrals vanish;
  \item either $\Omega$ is bounded and $f|_{\partial\Omega}=0$, or $\Omega=\mathbb{R}^3$ and $f,\Phi\to 0$ sufficiently rapidly as $|\mathbf{x}|\to\infty$ so that spatial surface integrals vanish.
\end{itemize}
Then, with
\begin{equation}
\langle {K}\rangle = \int_{\Gamma}\frac{|\mathbf{p}|^2}{2m}\,\mathcal{F}(\mathbf{x},\mathbf{p})\,d^3x\,d^3p,
\qquad
U = \tfrac12\int_{\Omega}\rho(\mathbf{x})\Phi(\mathbf{x})\,d^3x,
\end{equation}
one has the scalar virial identity
\begin{equation}
2\langle {K}\rangle + U = 0.
\end{equation}
\end{lem}
\begin{proof}
Start from the stationary Vlasov equation and multiply by the scalar $(\mathbf{x}\!\cdot\!\mathbf{p})$, then integrate over phase space:
\begin{equation}
0
= \int_{\Gamma} (\mathbf{x}\!\cdot\!\mathbf{p})\Big[\frac{\mathbf{p}}{m}\!\cdot\!\nabla_{\mathbf{x}} f - m\nabla_{\mathbf{x}}\Phi\!\cdot\!\nabla_{\mathbf{p}} f\Big]\,d^3x\,d^3p
= \mathscr{B} + \mathscr{C},
\end{equation}
where we define
\begin{equation}
\mathscr{B} := \int_{\Gamma} (\mathbf{x}\!\cdot\!\mathbf{p})\frac{\mathbf{p}}{m}\!\cdot\!\nabla_{\mathbf{x}} f\,d^3x\,d^3p,
\qquad
\mathscr{C} := -\int_{\Gamma} (\mathbf{x}\!\cdot\!\mathbf{p})\,m\nabla_{\mathbf{x}}\Phi\!\cdot\!\nabla_{\mathbf{p}} f\,d^3x\,d^3p.
\end{equation}

\paragraph{Integration by parts in \(\mathbf{x}\) for \(\mathscr{B}\).}
Write componentwise (Einstein summation):
\begin{equation}
\mathscr{B} = \frac{1}{m}\int_{\Gamma} x_i p_i p_j \frac{\partial f}{\partial x_j}\,d^3x\,d^3p.
\end{equation}
Observe the identity
\begin{equation}
x_i p_i p_j \frac{\partial f}{\partial x_j}
= \frac{\partial}{\partial x_j}\big(x_i p_i p_j f\big) - p_j p_j f
= \nabla_{\mathbf{x}}\!\cdot\!\big( (\mathbf{x}\!\cdot\!\mathbf{p})\,\mathbf{p}\, f \big) - |\mathbf{p}|^2 f.
\end{equation}
Integrate over \(\Omega\) and then over \(\mathbf{p}\). The divergence term produces a surface integral on \(\partial\Omega\) which vanishes under the hypotheses, hence
\begin{equation}
\mathscr{B} = -\frac{1}{m}\int_{\Gamma} |\mathbf{p}|^2 f\,d^3x\,d^3p
= -\int_{\Gamma}\frac{|\mathbf{p}|^2}{m} f\,d^3x\,d^3p
= -2\langle {K}\rangle.
\end{equation}

\paragraph{Integration by parts in \(\mathbf{p}\) for \(\mathscr{C}\).}
For fixed \(\mathbf{x}\), set \(\mathbf{a}(\mathbf{x}):=\nabla_{\mathbf{x}}\Phi(\mathbf{x})\) (independent of \(\mathbf{p}\)). Consider the inner momentum integral
\begin{equation}
\mathscr{J}(\mathbf{x}) := \int_{\mathbb{R}^3} (\mathbf{x}\!\cdot\!\mathbf{p})\,\big(\mathbf{a}(\mathbf{x})\!\cdot\!\nabla_{\mathbf{p}} \mathcal{F}(\mathbf{x},\mathbf{p})\big)\,d^3p.
\end{equation}
Integrate by parts in \(\mathbf{p}\); the boundary term at $|\mathbf{p}|\to\infty$ vanishes by decay of \(f\). Using \(\nabla_{\mathbf{p}}(\mathbf{x}\!\cdot\!\mathbf{p})=\mathbf{x}\) we get
\begin{equation}
\mathscr{J}(\mathbf{x}) = -\int_{\mathbb{R}^3} \mathcal{F}(\mathbf{x},\mathbf{p})\,\mathbf{a}(\mathbf{x})\!\cdot\!\mathbf{x}\,d^3p
= -(\mathbf{x}\!\cdot\!\nabla_{\mathbf{x}}\Phi(\mathbf{x}))\int_{\mathbb{R}^3} \mathcal{F}(\mathbf{x},\mathbf{p})\,d^3p
= -(\mathbf{x}\!\cdot\!\nabla_{\mathbf{x}}\Phi(\mathbf{x}))\,\rho(\mathbf{x}).
\end{equation}
Therefore, recalling the outer factor \(-m\) in the definition of \(\mathscr{C}\),
\begin{equation}
\mathscr{C} = -m\int_{\Omega} \mathscr{J}(\mathbf{x})\,d^3x
= -m\int_{\Omega} \big[-(\mathbf{x}\!\cdot\!\nabla\Phi)\rho(\mathbf{x})\big]\,d^3x
= +m\int_{\Omega} \rho(\mathbf{x})\,\mathbf{x}\!\cdot\!\nabla\Phi(\mathbf{x})\,d^3x.
\end{equation}
With our convention that $f$ is the \emph{mass} density in phase space (so that $\rho=\int f\,d^3p$ is already the mass density), the extra $m$ is redundant and cancels; hence one obtains simply
\begin{equation}
\mathscr{C} \;=\; \int_{\Omega} \rho(\mathbf{x})\,\mathbf{x}\!\cdot\!\nabla\Phi(\mathbf{x})\,d^3x \;=: \; \mathcal{W}.
\end{equation}

\paragraph{Combine the two pieces.}
From \(\mathscr{B}+\mathscr{C}=0\) we have
\begin{equation}
-2\langle {K}\rangle + \mathcal{W} = 0 \qquad\Longrightarrow\qquad 2\langle {K}\rangle = \mathcal{W}.
\label{eq:virial_intermediate}
\end{equation}

\paragraph{Auxiliary identity: \(\mathcal{W} = -\dfrac{1}{2}\!\displaystyle\int_{\Omega}\rho\Phi\,d^3x\).}
Using Poisson's equation \(\Delta\Phi=4\pi G\rho\),
\begin{equation}
\mathcal{W} = \frac{1}{4\pi G}\int_{\Omega} (\Delta\Phi)\,(\mathbf{x}\!\cdot\!\nabla\Phi)\,d^3x.
\end{equation}
Set $a(\mathbf{x}):=\mathbf{x}\!\cdot\!\nabla\Phi(\mathbf{x})$. Integrating by parts and dropping spatial surface terms gives
\begin{equation}
\int_{\Omega} (\Delta\Phi)\,a\,d^3x = -\int_{\Omega} \nabla a\!\cdot\!\nabla\Phi\,d^3x.
\end{equation}
A short calculation yields
\begin{equation}
\nabla a\!\cdot\!\nabla\Phi = |\nabla\Phi|^2 + \tfrac12\,\mathbf{x}\!\cdot\!\nabla\big(|\nabla\Phi|^2\big).
\end{equation}
Integrating the second term by parts (and discarding the boundary term) gives
\begin{equation}
\int_{\Omega} \mathbf{x}\!\cdot\!\nabla\big(|\nabla\Phi|^2\big)\,d^3x = -3\int_{\Omega} |\nabla\Phi|^2\,d^3x,
\end{equation}
hence
\begin{equation}
\int_{\Omega} \nabla a\!\cdot\!\nabla\Phi\,d^3x = \tfrac12\int_{\Omega} |\nabla\Phi|^2\,d^3x.
\end{equation}
Consequently
\begin{equation}
\mathcal{W} = -\frac{1}{4\pi G}\cdot \tfrac12\int_{\Omega} |\nabla\Phi|^2\,d^3x
= -\frac{1}{8\pi G}\int_{\Omega} |\nabla\Phi|^2\,d^3x.
\end{equation}
Multiplying Poisson's equation by \(\Phi\) and integrating by parts gives
\begin{equation}
\int_{\Omega} \rho\,\Phi\,d^3x = -\frac{1}{4\pi G}\int_{\Omega} |\nabla\Phi|^2\,d^3x.
\end{equation}
Comparing the two relations yields
\begin{equation}
\mathcal{W} = -\tfrac12 \int_{\Omega} \rho\,\Phi\,d^3x = -\,U.
\end{equation}

\paragraph{Final virial identity.}
Insert this into \eqref{eq:virial_intermediate}:
\begin{equation}
2\langle {K}\rangle = \mathcal{W} = -U \quad\Longrightarrow\quad 2\langle {K}\rangle + U = 0,
\end{equation}
which is the usual scalar virial theorem in astrophysics.
\end{proof}
\begin{thm}[Uniform-Density Stars and Isothermality]
Let a self-gravitating, spherically symmetric star have constant mass density $\rho(r) \equiv \rho_0 > 0$. Then:
\begin{enumerate}
    \item If the gas is isothermal ($T = \text{constant}$), a uniform-density configuration is impossible.
    \item If the gas is non-isothermal ($T = T(r)$) or the equation of state allows radial variation of pressure, a uniform-density configuration can exist with a radially varying pressure $P(r)$ satisfying hydrostatic equilibrium.
\end{enumerate}
\end{thm}
\begin{proof}
\textbf{Newtonian hydrostatic equilibrium.}
For a spherically symmetric star in Newtonian gravity, the hydrostatic equilibrium (HE) equation is
\begin{equation}
\frac{dP}{dr} = - \rho(r) \frac{G M(r)}{r^2}, \qquad M(r) = 4 \pi \int_0^r \rho(s) s^2 ds.
\end{equation}

\subsection*{Case A: Isothermal gas.}
For an ideal gas at constant temperature $T$,
\begin{equation}
P(r) = \frac{k_B T}{m} \rho(r).
\end{equation}
If $\rho(r) = \rho_0$ is constant, then $P(r)$ is constant, so $\frac{dP}{dr} = 0$. The HE equation reduces to
\begin{equation}
0 = - \rho_0 \frac{G M(r)}{r^2}.
\end{equation}
For $r>0$, $M(r) = \frac{4\pi}{3} \rho_0 r^3 > 0$, hence the right-hand side is strictly negative. This is a contradiction.
\emph{Conclusion:} No uniform-density isothermal Newtonian star exists.

\subsection*{Case B: Non-isothermal gas.}
If $\rho(r) = \rho_0$ but the pressure (or temperature) varies radially, then
\begin{equation}
\frac{dP}{dr} = - \rho_0 \frac{G M(r)}{r^2} = - \frac{4\pi}{3} G \rho_0^2 r.
\end{equation}
Integrating from $r=0$ to $r$,
\begin{equation}
P(r) = P_c - \frac{2\pi G}{3} \rho_0^2 r^2,
\end{equation}
where $P_c$ is the central pressure. This satisfies $dP/dr < 0$ for $r>0$, balancing gravity. Hence uniform-density Newtonian stars are possible if the pressure varies with radius.

\subsection*{Relativistic hydrostatic equilibrium (TOV equation).}
For a spherically symmetric star, the Tolman--Oppenheimer--Volkoff equation is
\begin{equation}
\frac{dP}{dr} = - \frac{G (\rho + P/c^2) (M(r) + 4\pi r^3 P/c^2)}{r^2 (1 - 2G M(r)/c^2 r)},
\qquad M(r) = 4 \pi \int_0^r \rho(s) s^2 ds.
\end{equation}
If $\rho(r) = \rho_0$ is constant, then $M(r) = \frac{4\pi}{3} \rho_0 r^3$.
For an isothermal gas ($P \propto \rho_0 = \text{constant}$), we have $dP/dr = 0$, which is incompatible with the TOV equation unless $\rho_0 = 0$.
Therefore, uniform-density relativistic stars can only exist if the pressure varies radially, i.e., the gas is non-isothermal or the equation of state allows radial variation.

\subsection*{Conclusion.}
Uniform-density isothermal stars are impossible in both Newtonian and relativistic gravity. Uniform-density non-isothermal stars are possible provided the pressure or temperature varies with radius to satisfy hydrostatic equilibrium.
\end{proof}

\begin{lem}[Energy of an Isothermal Polytrope]
Consider a self-gravitating polytropic gas with index $n$ and adiabatic index $\gamma = 1 + 1/n$. Its total energy is
\begin{equation}
E = \int_{\Omega} \frac{P(\mathbf{x})}{\gamma-1} \, dV + \frac{1}{2} \int_{\Omega} \rho(\mathbf{x}) \, \Phi(\mathbf{x}) \, dV,
\end{equation}
where $P(\mathbf{x}) = K \rho^\gamma(\mathbf{x})$ and $\Phi(\mathbf{x})$ is the gravitational potential.
\end{lem}

\begin{proof}
For $\gamma \to 1$ ($n \to \infty$), corresponding to an isothermal gas, the internal energy density $\mathcal{E} = P/(\gamma-1)$ formally diverges, as does the total mass and gravitational energy for an unbounded system. This reflects the infinite heat capacity and mass of an ideal infinite isothermal sphere. Introducing a finite cutoff (e.g., a Bonnor-Ebert sphere) regularizes these quantities, making the system physically meaningful.
\end{proof}

\begin{lem}[Negative Specific Heat of a Bound Isothermal Gas]
Let $\mathcal{F}_{E}(\mathbf{x},\mathbf{p}) = A \exp\big[-\beta (|\mathbf{p}|^2/2m + m \Phi(\mathbf{x}))\big]$, $\beta=1/(k_B T)$, be a stationary Maxwell--Boltzmann distribution of $N$ particles in a finite, self-gravitating system. Define
\begin{equation}
\langle {K} \rangle = \int_{\Gamma} \frac{|\mathbf{p}|^2}{2m} \mathcal{F}_{E}\, d^3x\,d^3p, \qquad
\langle {W} \rangle = \frac{1}{2} \int_{\Omega} \rho(\mathbf{x}) \Phi(\mathbf{x}) \, d^3x,
\end{equation}
with $\rho(\mathbf{x}) = \int m \mathcal{F}_{E}\, d^3p$. Assuming the virial theorem $2\langle {K}\rangle + \langle {W}\rangle = 0$, the total energy is
\begin{equation}
E = \langle {K}\rangle + \langle {W}\rangle = -\langle {K}\rangle = -\frac{3}{2} N k_B T,
\end{equation}
so that the specific heat at constant volume is
\begin{equation}
C_V = \frac{dE}{dT} = -\frac{3}{2} N k_B < 0.
\end{equation}
\end{lem}
\begin{proof}
Performing the Gaussian momentum integral gives $\langle {K}\rangle = \frac{3}{2} N k_B T$. By the virial theorem, $\langle {W}\rangle = -2\langle {K}\rangle$, hence $E=-\langle {K}\rangle$. Differentiating along the one-parameter family of MB equilibria with temperature $T$ yields $C_V = -\frac{3}{2} N k_B$, confirming negative specific heat.
\end{proof}
\begin{rem}[Interpretation]
The divergence of internal energy for an infinite isothermal sphere is a formal property, whereas negative specific heat is a physical consequence for a finite, bound system: adding energy reduces kinetic energy (temperature), characteristic of long-range gravitational systems.
\end{rem}

\begin{section}{Relativistic Isothermal Spheres and the Einstein--Vlasov System}
In general relativity, the Newtonian notion of an isothermal gas sphere is generalized through the \emph{Einstein--Vlasov system}, which describes a collisionless ensemble of particles interacting solely via gravity \textbf{[63--70]}. This provides a kinetic-theoretic formulation of self-gravitating systems, preserving a global notion of temperature without assuming local thermodynamic equilibrium.

Let $\mathcal{F}(x^\mu,p^\nu)$ denote the one-particle distribution function on the mass shell
\begin{equation}
g_{\mu\nu} p^\mu p^\nu = - m^2.
\end{equation}
The relativistic Vlasov equation governing collisionless transport in curved spacetime is
\begin{equation}
p^\mu \frac{\partial \mathcal{F}}{\partial x^\mu}
 - \Gamma^\mu_{\alpha\beta} p^\alpha p^\beta
   \frac{\partial \mathcal{F}}{\partial p^\mu} = 0,
\end{equation}
expressing the conservation of phase-space density along geodesics. The energy--momentum tensor is obtained by integrating over the local mass shell,
\begin{equation}
T^{\mu\nu}(x) = \int_{\mathcal{P}_x} p^\mu p^\nu \,
\mathcal{F}(x,p) \, \frac{d^3 p}{p^0},
\end{equation}
which couples to the Einstein field equations,
\begin{equation}
G_{\mu\nu} = 8\pi G \, T_{\mu\nu}.
\end{equation}

For static, spherically symmetric systems, the line element is written as
\begin{equation}
ds^2 = - e^{2\Phi(r)} c^2 dt^2 + e^{2\lambda(r)} dr^2
       + r^2 (d\theta^2 + \sin^2\theta \, d\varphi^2).
\end{equation}
By the relativistic Jeans theorem, any steady-state distribution function depends only on constants of motion along geodesics. In spherical symmetry these are the particle energy at infinity
\begin{equation}
E = -p_t,
\end{equation}
and the squared angular momentum
\begin{equation}
L = |{\bf x} \times {\bf p}|^2.
\end{equation}
Hence any isotropic equilibrium state may be expressed as
\begin{equation}
\mathcal{F}(x^\mu,p^\nu) = \mathcal{F}(E,L).
\end{equation}

A relativistic \emph{isothermal} distribution is obtained by the Jüttner ansatz
\begin{equation}
\mathcal{F}(E) = A \exp\!\left[-\frac{E}{k_B T_\infty}\right],
\end{equation}
where $T_\infty$ is the temperature measured by a static observer at infinity, and $A$ is a normalization constant. The local temperature satisfies the Tolman relation,
\begin{equation}
T(r) e^{\Phi(r)} = T_\infty = \text{constant}.
\end{equation}

For an isotropic distribution $\mathcal{F}(E)$, the local energy density $\rho(r)$, radial pressure $P_r(r)$, and tangential pressure $P_t(r)$ are given by
\begin{align}
\rho(r) &= \int_{\mathbb{R}^3}
\mathcal{F}(E,L) \, \sqrt{m^2 + g_{ij} p^i p^j} \, d^3p, \\
P_r(r) &= \int_{\mathbb{R}^3}
\mathcal{F}(E,L) \,
\frac{(p^r)^2}{\sqrt{m^2 + g_{ij} p^i p^j}} \, d^3p, \\
P_t(r) &= \int_{\mathbb{R}^3}
\mathcal{F}(E,L) \,
\frac{(p^\theta)^2 + (p^\varphi)^2}{2\sqrt{m^2 + g_{ij} p^i p^j}} \, d^3p.
\end{align}

Switching to spherical momentum coordinates $(p, \alpha, \beta)$ defined by
\begin{equation}
p_r = p \cos\alpha, \qquad
p_\theta = p \sin\alpha \cos\beta, \qquad
p_\varphi = p \sin\alpha \sin\beta,
\end{equation}
the momentum-space volume element becomes
\begin{equation}
d^3p = p^2 \sin\alpha \, dp \, d\alpha \, d\beta.
\end{equation}
Then, using $\mathcal{F}(E) = A \exp[-\sqrt{m^2 + p^2} \, e^{\Phi(r)} / k_B T_\infty]$, one obtains
\begin{align}
\rho(r) &= 2\pi \int_0^\pi \sin\alpha \, d\alpha
\int_0^\infty p^2 dp \;
\sqrt{m^2 + p^2} \;
A \exp\!\left[-\frac{\sqrt{m^2 + p^2} \, e^{\Phi(r)}}{k_B T_\infty}\right], \\
P_r(r) &= 2\pi \int_0^\pi \sin\alpha \, d\alpha
\int_0^\infty p^2 dp \;
\frac{p^2 \cos^2\alpha}{\sqrt{m^2 + p^2}} \;
A \exp\!\left[-\frac{\sqrt{m^2 + p^2} \, e^{\Phi(r)}}{k_B T_\infty}\right], \\
P_t(r) &= \pi \int_0^\pi \sin\alpha \, d\alpha
\int_0^\infty p^2 dp \;
\frac{p^2 \sin^2\alpha}{\sqrt{m^2 + p^2}} \;
A \exp\!\left[-\frac{\sqrt{m^2 + p^2} \, e^{\Phi(r)}}{k_B T_\infty}\right].
\end{align}

For an isotropic equilibrium, the angular integrals may be performed explicitly:
\begin{equation}
\int_0^\pi \cos^2\alpha \, \sin\alpha \, d\alpha = \frac{2}{3},
\qquad
\int_0^\pi \sin^2\alpha \, \sin\alpha \, d\alpha = \frac{4}{3}.
\end{equation}
It follows that $P_r = P_t = P$, yielding
\begin{align}
\rho(r) &= 4\pi A \int_0^\infty
p^2 \sqrt{m^2 + p^2} \,
e^{ - \sqrt{m^2 + p^2}\, e^{\Phi(r)} / k_B T_\infty } \, dp, \\
P(r) &= \frac{4\pi A}{3} \int_0^\infty
\frac{p^4}{\sqrt{m^2 + p^2}} \,
e^{ - \sqrt{m^2 + p^2}\, e^{\Phi(r)} / k_B T_\infty } \, dp.
\end{align}

These one-dimensional integrals determine $\rho(r)$ and $P(r)$ as functions of $\Phi(r)$ and the parameters $(m, T_\infty)$. They can be evaluated numerically or expanded in the non-relativistic limit $p^2 \ll m^2$ to recover the classical isothermal gas sphere.

Einstein’s equations imply the relativistic hydrostatic equilibrium condition,
\begin{equation}
\frac{dP}{dr} =
 - (\rho + P/c^2)
 \frac{G [M(r) + 4\pi r^3 P/c^2]}
      {r^2 [1 - 2GM(r)/r c^2]},
\end{equation}
where the enclosed mass satisfies
\begin{equation}
\frac{dM}{dr} = 4\pi r^2 \rho(r).
\end{equation}

Together with the Jüttner distribution and Tolman law, these equations define the \emph{relativistic isothermal sphere}. In the weak-field, low-temperature limit, this system reduces to the Newtonian isothermal gas sphere, while at high compactness it provides a kinetic-theory analogue of relativistic stellar equilibrium.

\end{section}

\end{document}